\documentclass[a4paper,UKenglish,cleveref, autoref,thm-restate,numberwithinsect]{lipics-v2021}
\nolinenumbers

\DOIPrefix{}

\usepackage{mathtools}
\usepackage{thmtools}
\usepackage{thm-restate}
\usepackage[utf8]{inputenc}
\usepackage{tikz}
\usetikzlibrary{trees,shapes, positioning}
\usepackage[appendix=append,bibliography=common]{apxproof}
\usepackage{relsize}
\usepackage{stackengine}
\stackMath{}
\usepackage{amsmath}
\usepackage{amssymb}
\usepackage{todonotes}
\usepackage{graphics,color}
\usepackage{hyperref}
\hypersetup{
  colorlinks,
  linkcolor={red!50!black},
  citecolor={blue!50!black},
  urlcolor={blue!30!black},
  breaklinks=true
  }
  \usepackage{booktabs}
  \usepackage{pgfplots}
  \pgfplotsset{compat=1.18}

\newtheoremrep{theorem}{Theorem}[section]

\newtheoremrep{proposition}[theorem]{Proposition}
\newtheoremrep{claim}[theorem]{Claim}
\newtheoremrep{corollary}[theorem]{Corollary}
\newtheoremrep{lemma}[theorem]{Lemma}
\newtheoremrep{example}[theorem]{Example}
\newtheoremrep{definition}[theorem]{Definition}

\newcommand{\NN}{\mathbb{N}}

\usepackage{xparse}
\usepackage{xspace}

\NewDocumentCommand{\Hcycle}{o}{\IfNoValueTF{#1}{\ensuremath{S}\text{-}Hamiltonian cycle\xspace}{\{#1\}\text{-}Hamiltonian cycle\xspace}}
\NewDocumentCommand{\Hpath}{o}{\IfNoValueTF{#1}{\ensuremath{S}\text{-}Hamiltonian path\xspace}{\{#1\}\text{-}Hamiltonian path\xspace}}
\NewDocumentCommand{\Hcyclebm}{o}{\IfNoValueTF{#1}{\ensuremath{\safebm{S}}\text{-}Hamiltonian cycle\xspace}{\{#1\}\text{-}Hamiltonian cycle\xspace}}
\NewDocumentCommand{\Hpathbm}{o}{\IfNoValueTF{#1}{\ensuremath{\safebm{S}}\text{-}Hamiltonian path\xspace}{\{#1\}\text{-}Hamiltonian path\xspace}}

\newcommand{\edge}[1]{\ensuremath{\{#1\}}}

\newcounter{sqindex}

\newcommand{\set}[1]{\{#1\}}

\newcommand{\safebm}[1]{\texorpdfstring{\bm{#1}}{#1}}

\usepackage{bm}

\title{The \texorpdfstring{$\bm{S}$}{S}-Hamiltonian Cycle Problem}

\author{Antoine Amarilli}{Univ. Lille, Inria, CNRS, Centrale Lille, UMR 9189
CRIStAL, F-59000 Lille, France}{a3nm@a3nm.net}{https://orcid.org/0000-0002-7977-4441}{}

\author{Arthur Lombardo}{DI ENS, ENS, PSL University, CNRS, Inria, Paris, France
  \and Univ. Lille, CNRS, Inria, Centrale Lille, UMR 9189
CRIStAL, F-59000 Lille, France}{arthur.lombardo@inria.fr}{https://orcid.org/0009-0006-1557-1932}{}

\author{Mikaël Monet}{Univ. Lille, Inria, CNRS, Centrale Lille, UMR 9189 CRIStAL, F-59000 Lille, France}{mikael.monet@inria.fr}{https://orcid.org/0000-0002-6158-4607}{}

\authorrunning{A. Amarilli, A. Lombardo and M. Monet}

\acknowledgements{This work was funded in part by the French government
    under management of Agence Nationale de la Recherche
    (ANR) as part of the “France 2030” program, reference
    ANR-23-IACL-0008 (PR[AI]RIE-PSAI).}

\Copyright{Antoine Amarilli, Arthur Lombardo, and Mikaël Monet}

\ccsdesc[500]{Mathematics of computing~Graph theory}

\keywords{Graph, Cycle, Hamiltonian}

\begin{document}

\maketitle

\begin{abstract}
Determining if an input undirected graph is Hamiltonian,
i.e., if it has a cycle that visits every vertex exactly once,
is one of the most famous NP-complete problems.
We consider the following generalization of Hamiltonian cycles:
for a fixed set $S$ of natural numbers, we want to visit
each vertex of a graph $G$
exactly once and ensure that any two consecutive
vertices can be joined in $k$ hops for some choice of $k \in S$.
Formally, an \emph{$S$-Hamiltonian cycle} is
a permutation $(v_0,\ldots,v_{n-1})$ of the vertices of~$G$
such that, for $0 \leq i \leq n-1$, there exists a walk between 
$v_i$ and $v_{i+1 \bmod n}$ whose length is in $S$.
(We do not impose any constraints on how many times vertices can be visited as intermediate
vertices of walks.)
Of course Hamiltonian cycles in the standard sense correspond to
$S=\{1\}$.
We study the \emph{$S$-Hamiltonian cycle problem}
of deciding whether an input graph $G$
has an $S$-Hamiltonian cycle. Our goal is to
determine the complexity of this problem
depending on the fixed set~$S$.
It is already known that the
problem remains NP-complete for $S=\{1,2\}$, whereas it is trivial for
$S=\{1,2,3\}$ because any connected graph 
contains a $\{1,2,3\}$-Hamiltonian cycle.

Our work classifies the complexity of this problem 
for most kinds of sets $S$, with
the key new results being the following: we have NP-completeness 
for $S = \{2\}$
and for $S = \{2, 4\}$, but tractability for $S = \{1, 2, 4\}$, for $S = \{2, 4,
6\}$, for any superset of these two tractable cases, and for $S$ the infinite set of all odd integers. The remaining open cases
are the non-singleton finite sets of odd integers,
in particular $S = \{1, 3\}$. Beyond cycles, we also discuss the complexity of finding
$S$-Hamiltonian paths, and show that our problems are all tractable on graphs of bounded cliquewidth.

 \end{abstract}

\section{Introduction}\label{sec:intro}

The Hamiltonian cycle problem is a well-known problem in graph theory and
complexity theory. It asks, given an undirected graph, whether it contains a
\emph{Hamiltonian cycle}, i.e., a cycle that visits every vertex exactly once.
This problem is well-known to be NP-complete, and the same is true if we want to
find a \emph{Hamiltonian path}, i.e., a path that visits every vertex exactly
once.

Faced with the hardness of this problem, one natural relaxation is to allow
greater distance bounds between pairs of consecutive vertices instead of
requiring them to be adjacent. The simplest
case is to allow hops of length $1$ or $2$ (instead of just 1).
Equivalently, we consider the
\emph{square} $G^2$ of the input graph $G$, in which we connect any two vertices
that can be joined by a path of length at most~$2$, and we ask whether $G^2$ is
Hamiltonian. It is still NP-complete to determine whether an input graph admits
such a relaxed cycle~\cite{On_graphs_with_Underg_1978}, even though many graph
families are known whose square is always 
Hamiltonian: 
this is the case of two-connected
graphs~\cite{The_square_of_e_Fleisc_1974,muttel2013short} and
other
families~\cite{hendry1985square,abderrezzak1999induced,ekstein2012hamiltonian,ekstein2021best,ekstein2024most}.
(We note that there are also other different (and seemingly unrelated) ways to define
generalizations of Hamiltonicity, e.g., taking the
iterated line graph as in~\cite{chartrand1968hamiltonian,sabir2014spanning}.)

The natural next question is to allow hops of 1, 2, or 3. However, it is
then known that the cube $G^3$ of every connected
graph $G$ is Hamiltonian~\cite{karaganis1968cube,sekanina1960ordering}, and this result is a commonly used trick
for efficient enumeration algorithms
(the so-called ``even-odd trick''~\cite{karaganis1968cube,sekanina1960ordering} (see also~\cite{uno2003two}). The case of hops of length at most~$3$
immediately implies that relaxed Hamiltonian cycles always exist in connected
graphs for any other set of the form $S = \{1, \dots, k\}$ with $k > 3$.

However, this does not address the question of other fixed sets of allowed hop
lengths. For instance, we can consider $S = \{2\}$: this asks whether we can
find a permutation $(v_0,\dots,v_{n-1})$ of the vertices of the input graph~$G$ in which any two
consecutive vertices (including the endpoints $v_{n-1}$ and $v_0$) can be joined
by a walk of length \emph{exactly} two. 
We call such a cycle a \emph{\Hcycle[2]}.
This question relates to the Hamiltonicity of the
\emph{common neighborhood graph} of~$G$~\cite{On_the_common_n_Hossei_2014}, i.e., the graph where we connect pairs of vertices that
are joined by such walks. 
Hossein et al.~\cite[Theorem~4.1]{On_the_common_n_Hossei_2014} 
in particular give a sufficient condition for graphs to have a \Hcycle[2], but
they know that the condition is not necessary~-- so to our knowledge there is no
characterization of the graphs admitting a \Hcycle[2] and no complexity bounds
on the problem of recognizing them. Alternatively, one can consider, say, $S = \{1, 2, 4\}$: do all graphs admit a \emph{$\set{1,2,4}$-Hamiltonian
cycle}, i.e., a cycle where any two consecutive vertices are joined by a walk
with some length in~$S$?

\subparagraph*{Contributions.}
In this paper we study which graphs admit $S$-Hamiltonian cycles, and the
complexity of recognizing them, depending on the fixed set~$S$. We show that,
depending on the choice of~$S$, the problem behaves very differently:
there are sets~$S$ for which every (connected) graph admits an \Hcycle, other
cases where the
decision problem is NP-hard, and last some cases where the decision problem is
non-trivial but efficiently solvable. We entirely classify these behaviors,
except for one remaining family of open cases.

The main results come from the sets $S = \set{2}$, $S = \set{2,4}$, $S = \set{2,4,6}$ and $S = \set{1,2,4}$.
For $S = \set{2}$, we prove 
that the \Hcycle[2] problem
is NP-complete. This proceeds by a reduction from the classical \Hpath[1]
problem: we reduce it to the \Hpath[2] problem with specified endpoints by building the incidence
graph of the input graph and attaching a triangle, and we then reduce that problem to the \Hcycle[2] problem by attaching a specific gadget.
For $S = \set{2,4}$, we show that the \Hcycle[2,4] problem is NP-hard under Cook
reductions.
Specifically, we reduce from the \Hcycle[1,2] problem (which is known to be
NP-hard~\cite{On_graphs_with_Underg_1978})
and build graphs 
by attaching a specific gadget on every pair of sufficiently close vertices.
As for $S = \set{2,4,6}$, we show that a connected graph admits an \Hcycle[2,4,6] if and
only if it is non-bipartite: this gives a linear-time recognition algorithm, and
we show that a witnessing cycle can also be built in linear time when one exists.
Finally, for $S = \set{1,2,4}$, we prove that every connected graph admits a
\Hcycle[1,2,4] by first showing the same result on trees. The proof is by induction
on the number of vertices, and it also yields a linear-time algorithm to build a
witnessing cycle.

The remaining family of unclassified cases are the non-singleton finite sets of
odd integers, i.e., the sets $S = \set{1,\dots,2k+1}$ for $k>1$, for which
the complexity of the \Hcycle{} problem remains open.

We also discuss three variants of the \Hcycle{} problem, the first being the variants when the set $S$ is infinite,
which are all tractable. The second variant is the study of the complexity of \Hpath{}, i.e., we study the complexity of deciding
whether an input graph contains an \Hpath{} (possibly with specified endpoints);
we show that in most cases this has the same complexity as the \Hcycle{}
problem. Finally, we also study the tractability of the \Hcycle{} problem 
when restricting the graphs that are allowed as inputs: we show that, for any
fixed finite set~$S$, the problem is tractable when the input graphs are
required to have bounded cliquewidth

\subparagraph*{Paper outline.}
In \cref{sec:prelim}, we introduce the definitions and notations that are used in this paper,
and give more formal details about the \Hcycle{}
problem and its variants.
In \cref{sec:result}, we present our main result and a decision tree that delineates
the tractable, intractable, and open cases. We then give the detailed proofs for
each case. First,
\cref{sec:np} is dedicated to the proofs of the NP-completeness of the \Hcycle[2] and \Hcycle[2,4] problems.
Then, \cref{sec:trivial} deals with the \Hcycle[1,2,4] problem, and \cref{sec:poly} with the \Hcycle[2,4,6] problem.
We then discuss problem variants in \cref{sec:variants}. More specifically,
in \cref{sec:infinite}, we study the complexity of the \Hcycle{} problem when
$S$ is an infinite set.
We discuss the complexity for the \Hpath{} variants of our problem in
\cref{sec:pathandse,sec:124sehpath}.
Finally, in \cref{sec:restricted}, we study the case of graphs of bounded cliquewidth.
We conclude in \cref{sec:conc}.
For lack of space, most detailed proofs are deferred to the appendix.

\section{Preliminaries and Problem Statement}\label{sec:prelim}

\subparagraph*{Standard graph notions.}
All graphs in this paper are undirected, do not feature self-loops, are non-empty
(i.e., have at least one vertex), and are simple
graphs (i.e., without parallel edges).
Furthermore, we always assume that graphs are connected.
Formally then, a \emph{graph} $G = (V, E)$ consists of a finite set of vertices $V$ (also denoted
$V(G)$) and a set of edges $E$ (also denoted $E(G)$) which is a set of pairs of vertices,
i.e., the edge \edge{u,v} connects the distinct vertices $u$ and $v$, and we also say
the edge is \emph{incident} to~$u$ and to~$v$, and that $u$ and $v$ are
\emph{adjacent}.
A \emph{spanning subgraph}
of~$G$ is a graph $(V, E')$ with $E' \subseteq E$ that is connected.
A \emph{tree} is an acyclic graph, and it is a \emph{rooted tree} if a vertex has been designated as the root.
A \emph{spanning tree} $T$ of $G$ is a spanning subgraph of $G$ which is a tree.

A \emph{walk} in a graph $G$ is a sequence of vertices
$(v_0, v_1, \dots, v_k)$ such that for each $i \in \{0, \dots, k-1\}$
the vertices $v_i$ and $v_{i+1}$ are adjacent.
The \emph{length} of the walk is defined as the number of edges traversed, which is \(k\).
A \emph{simple path} in a graph $G$ is a walk where all vertices are distinct, and a
\emph{simple cycle} is a simple path where the first and the last vertices are adjacent.
The length of a cycle is then defined as one more than the length of its defining walk.
We require graphs to be \emph{connected}, i.e., for any choice of two vertices
there exists a walk that connects them.

The \emph{incidence graph} of $G = (V, E)$ is the graph $G' = (V', E')$
whose vertices are the vertices and edges of~$G$ and where each edge of $G$ is connected in $G'$ to its incident
vertices; formally $V' \coloneq V \uplus E$ and $E' \coloneq \{\{u, e\} \mid e
\in E, u \in e\}$.
The \emph{line graph} of~$G$, denoted $L(G) = (V', E')$, is the graph on
the edges of~$G$ where two edges of $G$ are connected in $G'$ if they are incident in $G$ to a
common vertex; formally $V' \coloneq E$ and $E' \coloneq \{\{e,e'\} \mid e, e'
\in E, e \cap e' \neq \emptyset\}$. 
We call $G$ \emph{bipartite} if its vertex set can be partitioned into two
non-empty sets $X$ and $Y$, called the \emph{parts}, such that each edge of $G$ connects a vertex in $X$
to a vertex in $Y$. We recall that a graph is bipartite if and only if it does
not admit an odd-length simple cycle,
and that checking whether an input graph is
bipartite can be done in linear time.

\subparagraph*{Problem definition.}
We now formally define the main notions studied in this paper. In the following, we assume $S$ is a non-empty finite subset of $\NN^+$, the positive integers (we will study the case of infinite $S$ in \cref{sec:infinite}).

\begin{definition}[\textbf{\safebm{$S$}-path}]\label{def:spath}
  Let $G = (V, E)$ be a graph, and let $S \subseteq \NN^+$.
  A sequence of vertices $P = (v_1, v_2, \dots, v_n)$ is called an \emph{$S$-path} of $G$ if,
  for each $i \in \{1, \dots, n-1\}$, the vertices $v_i$ and $v_{i+1}$ are
  connected by a walk of length $\ell \in S$, i.e.,
  there exists some length $\ell \in S$ and some walk in $G$
  of length $\ell$ that starts at $v_i$ and ends at $v_{i+1}$.
  Furthermore,~$P$ is called a \emph{simple $S$-path} if all its
  vertices are distinct.
\end{definition}

\begin{definition}[\textbf{\Hpathbm} and \textbf{\Hcyclebm}]\label{def:hpath}
    Let $G = (V, E)$ be a graph, and let $S \subseteq \NN^+$.
    An~\emph{\Hpath} of~$G$ is a simple $S$-path that contains all the vertices of $G$.
    An~\emph{\Hcycle} of~$G$ is an \Hpath{} whose first and last vertices are connected by a walk of a length $\ell \in S$ in $G$.
\end{definition}

For example, the classical Hamiltonian cycle problem corresponds to
the case $S = \{1\}$, as two vertices are connected by a walk of length $1$
precisely when they are adjacent.

\subparagraph*{Alternative definitions.}
Note that our definitions above require the
existence of a \emph{walk} of a length in $S$ connecting any two consecutive
vertices. The rest of the paper only considers this choice of definition, but we briefly
discuss here some alternative possible choices, in order to dispel
potential confusion.

First, we do \emph{not} require that the witnessing walk is a \emph{shortest path}. In
other words, we do not require that the \emph{distance} between two consecutive
vertices $u$ and $v$ belongs to~$S$, only that \emph{some possible walk} has length in~$S$: there may
be shorter walks connecting~$u$ and~$v$ with a length not
in~$S$. 
Requiring shortest paths as connecting walks would be a
more stringent requirement in general, even though there are some sets $S$
for which both definitions coincide (e.g.,
sets of the form $S = \{1,\ldots,k\}$ for some $k > 0$).
We note that, already for the set $S = \{2\}$,
requiring shortest paths between consecutive vertices would be different,
and it would relate to the Hamiltonicity of the 
\emph{exact-distance square} of the input
graph~\cite{Characterizing_Bai_Y_2024}. We do not study this question in the
present work.

Second, we do \emph{not} require that the witnessing walks between consecutive vertices
are \emph{simple paths}. Requiring simple paths would again be a more stringent
requirement in general~--though again it would make no difference for sets of the form
$S = \{1,\ldots,k\}$, and it would also make no difference for
$S = \{2\}$. Again, in this work we do not study this alternative definition
requiring simple paths as connecting walks.

Third, we reiterate that our definitions do not pose any restriction on the
number of times that a vertex may occur in the connecting walks. In particular,
note that vertices that have already been visited earlier in the permutation can still be traversed as intermediate vertices of later connecting walks.

Last, we point out one important consequence of our choice to allow
connecting walks that may feature repeated vertices: when two vertices are
connected by a walk of length $\ell>0$, then they are also
connected by a walk of length $\ell + 2$, simply by going back-and-forth on the
last edge.
Thus, in our problem, we can always assume without loss of generality that the
set $S$ is \emph{closed under subtraction of $2$}, i.e., whenever $\ell \in S$ is an
allowed distance and $\ell-2 > 0$ then the distance $\ell-2$ is also allowed.
In other words, whenever $S$ contains a number $\ell$, then we
assume that it also contains all smaller numbers of the same
parity -- for clarity we will always write down these numbers explicitly.

\subparagraph*{Computational complexity.}
The main problem we study is the \Hcycle{} problem for fixed sets $S$, which asks,
given a graph $G$, whether $G$ admits an \Hcycle. We also study in
\cref{sec:variants} the complexity of two variants
the \emph{\Hpath{} problem}, for which we need to decide whether $G$ has an \Hpath, and
the
\emph{\Hpath problem with specified endpoints}, which asks whether 
$G$ admits an \Hpath{}
that starts and ends at some specified endpoints given as input.
For all these three problems, the complexity is always measured as usual as a function
of the input graph. We note that these problems are always in NP for any
choice of fixed~$S$: the certificate is the permutation of vertices $s = (v_1,
\ldots, v_n)$, and it is easy to check in polynomial time that the certificate
is correct. We will show cases when the problem is NP-hard, and other cases when
the problem can be efficiently solved: either because it is trivial (i.e., all
graphs have an \Hcycle{}), or because there is an efficient algorithm to decide
it. In the efficient cases, we will also study the complexity of the problem of
efficiently finding a witness (e.g., an \Hcycle{}) when one exists, and give tractable algorithms for this task.

For some problems, we will only be able to show NP-hardness
under \emph{Cook reductions} (i.e., polynomial-time Turing reductions), rather than the more commonly used 
\emph{Karp reductions} (i.e., polynomial-time many-one reductions).
Recall that a many-one reduction is a 
reduction that maps each instance of a problem to exactly one instance of
another problem
whereas a Turing reduction is allowed to use the oracle of the target problem multiple times on different instances.
Thus, saying that a problem is NP-complete under Cook reductions is a weaker statement than saying that it is NP-complete under Karp reductions,
but it still implies that there is no polynomial-time
algorithm for the problem unless P = NP.

\section{Main Result}\label{sec:result}
Having formally defined our various problems, we now present in
this section the main result of this paper, which classifies the
complexity of the \Hcycle{} problem, up to some remaining open
cases.  Recall that sets $S$ are considered to be finite, unless
stated otherwise~-- we revisit this choice in \cref{sec:variants}.
Results on the other variants (\Hpath, with or
without specified endpoints), and on the generalization of our
problems when $S$ can be infinite, will also be presented in
Section~\ref{sec:variants}. 

\begin{theoremrep}\label{thm:main}
  For every finite and non-empty set $S \subseteq \mathbb{N}^+$
the \Hcycle{}-problem is either
  NP-complete under Cook reductions, in P, trivial (i.e., true on all graphs), or open, as depicted in \cref{fig:dectree}.
  Further, when the problem is in P or trivial,
  we can compute a witnessing \Hcycle{} in linear time in the input graph.
\end{theoremrep}

We explain in the remaining of this section the roadmap to prove \cref{thm:main}. We start by recalling the relevant known results from the literature.

\begin{proposition}[\cite{On_the_Computat_Karp_1975}]\label{prop:1hcycle}
    The \Hcycle[1] problem is NP-complete.
\end{proposition}

\begin{proposition}[\cite{On_graphs_with_Underg_1978}]\label{prop:12hcycle}
    The \Hcycle[1,2] problem is NP-complete.
\end{proposition}

The following will be particularly useful to us:

\begin{proposition}[\cite{karaganis1968cube,sekanina1960ordering}]\label{prop:123hcycle}
 The  \Hcycle[1,2,3] problem is trivial. More precisely,
    for every graph $G$ and $u,v \in V(G)$ with $u\neq v$, there
    exists a \Hpath[1,2,3] from $u$ to $v$ in $G$.
\end{proposition}

Now, we present the new results that complete the classification of \cref{thm:main}.

\begin{restatable}{theorem}{twohcycle}\label{thm:2hcycle}
 The \Hcycle[2] problem is NP-complete.
\end{restatable}
\begin{restatable}{theorem}{twofourhcycle}\label{thm:24hcycle}
 The \Hcycle[2,4] problem is NP-complete under Cook reductions.
\end{restatable}
\begin{restatable}{theorem}{onetwofourhcycle}\label{thm:124hcycle}
Every connected graph admits a \Hcycle[1,2,4].
Furthermore, such a cycle can be found in linear time.
\end{restatable}
\begin{restatable}{theorem}{twofoursixhcycle}\label{thm:246hcycle}
The graphs that have a \Hcycle[2,4,6] are exactly the non-bipartite graphs,
which implies that the \Hcycle[2,4,6] problem can be solved in linear time.
Furthermore, on graphs having a \Hcycle[2,4,6], we can construct one in linear time.
\end{restatable}

From these results, we can decide the complexity of any set $S$ thanks to the
decision process illustrated by the decision tree in \cref{fig:dectree}, where
we assume without loss of generality that $S$ is closed under subtraction
of~$2$.
In this decision tree, we denote by \emph{P} that the problem corresponding to
the sets $S$ considered is polynomial, by \emph{NP-c} that it is NP-complete
(possibly under Cook reductions),
by \emph{Trivial} that every graph admits such a cycle, and by \emph{open} that we do not know yet. The new results are in the bold nodes.
The detailed proof of the correctness of that decision process, which is also a
proof of \cref{thm:main}, is in Appendix~\ref{apx:main}.
That said, the only non-obvious reasoning in this proof is the following.
First, for sets containing \(6\), we know by \cref{thm:246hcycle} that the lengths \(2\), \(4\), and \(6\) are
enough to find a \Hcycle{} in any non-bipartite graph. On the other hand, having more even lengths will never allow to get a \Hcycle{} in a bipartite graph for obvious reasons, thus sets $S$ with only even numbers that are supersets of \set{2,4,6} can be grouped together.
Second, all sets $S$ that are supersets of \set{1,2,3} can be grouped together because
we know by \cref{prop:123hcycle} that the lengths
\(1,2\) and \(3\) are enough to find an \Hcycle{} in any graph. Third, all sets
$S$ that are supersets of \set{1,2,4} can be grouped together for the same reason thanks to \cref{thm:124hcycle}.

\begin{toappendix}
  \label{apx:main}
\begin{proof}
  Given a finite and non-empty set $S$, we assume without loss of generality
  that it is closed by subtraction of $2$. We then distinguish three cases:
  either $S$ contains only even numbers, or it contains even and odd numbers, or it contains only odd numbers.
  
  Case 1: $S$ contains only even numbers.
  Then there are two cases: either $S$ contains $6$ or it does not.
  First, if $6 \in S$,
  then the characterization is that an \Hcycle{} exists in a given graph $G$ 
  if and only if $G$ is non-bipartite. Specifically, if $G$ is bipartite, then
  the allowed lengths in $S$ will not allow us to move from one part of the
  bipartition to the other, so there is no \Hcycle{}. Conversely, if $G$ is
  not bipartite, then
we show an algorithm in \cref{thm:246hcycle} that efficiently creates a \Hcycle[2,4,6], which is also an \Hcycle{}.
  Second, if $6 \notin S$,
then we have either $S = \set{2}$ or $S = \set{2,4}$
  and both cases are proved to be NP-complete in \cref{thm:2hcycle} and \cref{thm:24hcycle} respectively.
  
  Case 2: $S$ contains even and odd numbers.
  In this case, $S$ must contain $1$ and $2$. If those are the only elements of $S$, then the \Hcycle[1,2] problem is NP-complete, as stated in \cref{prop:12hcycle}.
  Otherwise, $S$ must additionally contain either $3$ or $4$, in which case every graph contains an \Hcycle.
  Indeed, if $S$ also contains $3$, we saw in \cref{prop:123hcycle} that every graph admits a \Hcycle[1,2,3],
  which must also be an \Hcycle.
  Otherwise, if $S$ also contains $4$, then we prove in \cref{thm:124hcycle} that every graph admits a \Hcycle[1,2,4] which is an \Hcycle{} as well.

  Case 3: $S$ contains only odd numbers.
  Then if $S = \{1\}$ we know that the problem is NP-complete as stated in \cref{prop:1hcycle}.
  Otherwise, $S = \set{1,\dots,2k+1}$ with $k>1$ and this is the case that our
  work leaves open.
\end{proof}
\end{toappendix}

\begin{figure}[ht]
  \centering
  \begin{tikzpicture}[
    level 1/.style={sibling distance=4.6cm, level distance=2cm},
    level 2/.style={sibling distance=2.3cm, level distance=3cm},
    level 3/.style={level distance=2.5cm},
    question/.style={ellipse, draw},
    answer/.style={rectangle, draw, align=center},
    new/.style={ultra thick},
    xscale=1.05,
    edge from parent path={(\tikzparentnode.south) -- (\tikzchildnode.north)}
    ]
      
    \node[question] {Parity of elements in $S$?}
        child { node[question] {$6 \in S$?}
            child {node[answer, new] {\textbf{P} \\ $S = \set{2,4,6}$ \\ Thm~\ref{thm:246hcycle}}
            edge from parent node[left] {yes}}
            child {node[answer, new] {\textbf{NP-c} \\ $S = \set{2}$ \\ Thm~\ref{thm:2hcycle}, \\ $S = \set{2,4}$ \\ Thm~\ref{thm:24hcycle}}
            edge from parent node[right]{no}}
        edge from parent node[left, xshift=-5pt, yshift=5pt] {Only even numbers}}
        child { node[question] {$4 \in S$?} 
            child {node[answer, new] {\textbf{Trivial} \\ $S = \set{1,2,4}$ \\ Thm~\ref{thm:124hcycle}}
            edge from parent node[left] {yes}}
            child {node[question] {$3 \in S$?}
                child {node[answer] {Trivial \\ $S = \set{1,2,3}$ \\ See~\cite{karaganis1968cube,sekanina1960ordering}} edge from parent node[left] {yes}}
                child {node[answer] {NP-c \\ $S = \set{1,2}$ \\ See~\cite{On_graphs_with_Underg_1978}} edge from parent node[right] {no}}
            edge from parent node[right]{no}}
        edge from parent node[right]{Both}}
        child { node[question] {$3 \in S$?} 
          child {node[answer] {\emph{Open} \\ $S = \set{1,3,\ldots,k}$ \\ for some odd $k$}
            edge from parent node[left] {yes}}
            child {node[answer] {NP-c \\ $S = \set{1}$ \\ See~\cite{On_the_Computat_Karp_1975}}
            edge from parent node[right]{no}}
        edge from parent node[right, xshift=5pt, yshift=5pt]{Only odd numbers}};
      
    \end{tikzpicture}
    \caption{Decision tree for the complexity of the $S$-Hamiltonian cycle
    problem for $S$ a finite and non-empty subset of $\mathbb{N}^+$.}\label{fig:dectree}
\end{figure}
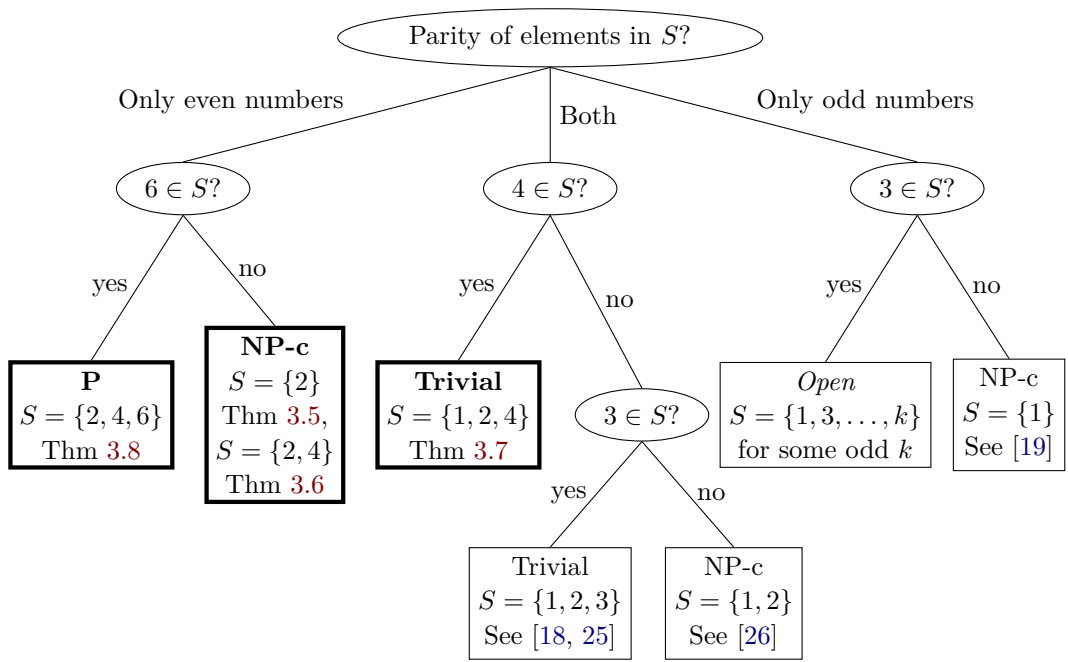

The next three sections are dedicated to proving the four main theorems of this section.

\section{NP-Complete Cases}\label{sec:np}
In this section, we start the proof of \cref{thm:main} by giving the NP-hardness
results for the cases $S = \{2\}$ (in \cref{sec:2}) and $S = \{2,4\}$ (in
\cref{sec:24}).

\subsection{The Case \texorpdfstring{$\bm{S = \{2\}}$}{S = \{2\}}}\label{sec:2}
\begin{toappendix}
  \subsection{Proofs for \texorpdfstring{$\bm{S = \{2\}}$}{S = \{2\}}}
  \label{apx:2}
\end{toappendix}

To show that the \Hcycle[2] problem is NP-hard, we proceed in two steps. We first reduce 
the \Hpath[1] problem with specified endpoints which is known to be NP-complete \cite{On_the_Computat_Karp_1975}, to the \Hpath[2] problem with specified endpoints.
Then, we reduce the latter problem to the \Hcycle[2] problem.
Let us show the first step:

\begin{theoremrep}\label{thm:2sehpath}
    The \Hpath[2] problem with specified endpoints is NP-complete.
\end{theoremrep}

\begin{proofsketch}
    Let $G$ be a graph and $\alpha \neq \beta$ be the specified endpoints in~$V(G)$ for 
    the Hamiltonian path problem.
    Construct a new graph~$H$ by first taking the incidence graph of~$G$ (which
    features an \emph{edge-vertex} for each edge of~$G$) and then
    connecting $\beta$ to a dangling triangle. Pick an arbitrary edge-vertex
    $\beta'$ adjacent to~$\alpha$ in~$H$.
    We can then conclude by establishing the following:
    $G$ has a \Hpath[1] connecting $\alpha$ and $\beta$ if and only if~$H$
    has a \Hpath[2] connecting $\alpha$ and $\beta'$.

    In particular, for the forward direction, to obtain a \Hpath[2] in $H$ we first follow the original \Hpath[1] of $G$, which is now a simple $\{2\}$-path from $\alpha$ to $\beta$ in $H$,
    then goes through the triangle to reach an edge-vertex.
    From there, we use a strengthening of~\cite[Theorem~6.5.4]{A_Textbook_of_G_Balakr_2012} that states that
    the line graph of a Hamiltonian graph is Hamiltonian, which allows us to visit all edge-vertices.
\end{proofsketch}

\begin{proof}
    We prove the claim by reduction from the \Hpath[1] problem with
    specified endpoints, which is NP-hard~\cite{On_the_Computat_Karp_1975}.

    Let $G$ be the input graph, and let $\alpha \neq \beta$ be the specified
    endpoints.
    We can assume without loss of generality that
    $G$ has at least two vertices, as otherwise the answer is trivial.
    Without loss of generality, up to modifying $G$, we can assume that $\alpha$ and $\beta$ are not adjacent.
    Indeed, if they are, we can add a new vertex to $G$ connected only to $\beta$ and make it the new $\beta$,
    which does not change the existence of a \Hpath[1] in $G$.

    We will build in polynomial time a new graph $H$ and choose two vertices $x
    \neq y$ in~$V(H)$ so as to ensure the following equivalence (*):
    $G$ has a \Hpath[1] from $\alpha$ to $\beta$ if and only if
    $H$ has a \Hpath[2] from $x$ to $y$.
    The construction of $H$ is as follows:
    Take the incidence graph $G'$ of~$G$ and connect $\beta$
to a triangle 
(\cref{fig:gadget2sehpath}).
    \begin{figure}[ht]
        \centering
        \begin{tikzpicture}
            \node[circle, draw] (d1) at (0,0) {$d_1$};
            \node[circle, draw] (d2) at (2,0) {$d_2$};
            \node[circle, draw] (d3) at (1,1.44) {$d_3$};

            \draw (d1) -- (d2);
            \draw (d2) -- (d3);
            \draw (d3) -- (d1);
        \end{tikzpicture}
        \caption{Triangle gadget $D$ used in the proof of \cref{thm:2sehpath}}\label{fig:gadget2sehpath}
    \end{figure}
    Formally:

    \begin{itemize}
      \item Build the incidence graph $G'$ by doing the following:
        \begin{itemize}
        \item For each vertex $v \in V(G)$, add a corresponding vertex $v' \in
          G'$. We call these \emph{node-vertices}.
        \item For each edge $\edge{u,v} \in E(G)$, add a new vertex $w_{uv}$ to
          $G'$ and connect it to both $u'$ and $v'$.
          We call these \emph{edge-vertices}.
        \end{itemize}
      \item Let $H$ be $G'$ together with a copy of the triangle gadget $D$
        from \cref{fig:gadget2sehpath} and an edge $\edge{\beta',d_1}$
        connecting $d_1$ to the vertex $\beta'$ of~$H$ that corresponds
        to~$\beta$ in~$G$.
    \end{itemize}
    Last, define $x \coloneq \alpha'$, and define $y$ to be any edge-vertex
    adjacent to $x$ in $H$. (Such an edge-vertex necessarily exists because the
    input graph $G$ is connected so $\alpha$ must have some
    incident edge in~$G$.)

    The construction of~$H$ given above is clearly in polynomial time, so all
    that remains is to establish correctness, i.e., proving the equivalence (*)
    stated above.
    Before diving into the proof, we will first see which pairs of vertices are connected by a walk of length $2$ in $H$.
    One can easily check that they are precisely the following:
    \begin{itemize}
        \item Node-vertices whose originals are adjacent in $G$: 
        $u', v' \in V(H)$ where $uv \in E(G)$, with the walk $u' - w_{uv} - v'$ of length $2$.
        \item Edge-vertices whose original edges share a common vertex: $w_{uv}, w_{vt}$ with $u, v,
          t \in V(G)$ and $uv, vt \in E(G)$, via the intermediate vertex $w_{uv} - v' - w_{vt}$.
        \item Special cases thanks to $D$:
          \begin{itemize}
            \item all vertices of $\{d_1, d_2, d_3\}$ are connected  to each
              other;
            \item $\beta'$ is connected to both $d_2$ and $d_3$;
\item $d_1$ is connected (through $\beta'$) to every edge-vertex corresponding
        to an edge of $G$ incident to $\beta$.
    \end{itemize}
    \end{itemize}

    We will now show equivalence (*). 
We first prove the forward direction.
Suppose that $G$ contains a \Hpath[1]
    $P = (v_1, v_2, \dots, v_k)$ such that $\alpha = v_1$ and $\beta = v_k$.
    We need to prove that $H$ contains a \Hpath[2] from $x$ to~$y$.
    To achieve this, we will construct a \Hpath[2] $P'$ in $H$ by starting from the node-vertices corresponding
    to the vertices in $P$ and then adding the vertices of $D$ and the edge-vertices. We start by initializing $P' = (v_1', \dots, v_k')$.
    We next add the vertices of $D$, which will allow us to reach an edge-vertex: $P' = (v_1', \dots, v_k', d_2, d_3, d_1)$.

    To complete the definition of~$P'$, we will now sort the edge-vertices
    into sets depending on their latest adjacent node-vertex according to $P$.
    We will then add all those sets in reverse order of their latest adjacent node-vertex in $P$.
    This is possible because the line graph of a Hamiltonian graph is also Hamiltonian~\cite{A_Textbook_of_G_Balakr_2012},
    and because the graph that connects two edge-vertices if they share a common node-vertex is exactly the line graph of $G$.
    This process is a stronger version of the one proved in~\cite[Theorem~6.5.4]{A_Textbook_of_G_Balakr_2012} because it further
    allows us to choose our endpoints in the line graph.

    Formally, we first create the sets $P_{i}$ for $i \in \set{3,\dots,k}$.
    Then, because $P$ is Hamiltonian, it must contain all node-vertices.
    This means that every edge-vertex is of the form $w_{v_i v_j}$ where $v_i, v_j \in P$.
    Thus, for each edge-vertex $w_{v_i v_j}$ except for $y$ and $w_{v_1 v_2}$, we add it to the set $P_{\max(i,j)}$.
    Now, all the $P_i$ are disjoint and their union contains all edge-vertices of $H$ except $y$ and $w_{v_1 v_2}$.
    Furthermore, notice that for all $i \in \set{3, \dots, k}$, $w_{v_{i-1} v_i} \in P_i$.

    We can now add those sets of edge-vertices to $P'$ in this order:
    For each $i$ from $k$ down to $3$, add the vertices of $P_i$ to $P'$ by adding $w_{v_{i-1}v_i}$ last.
    Then, if $y \neq w_{v_1 v_2}$, add $w_{v_1 v_2}$ and then add $y$ to $P'$.
    Else, add $y$ to $P'$.

    The reader can then easily check that $P'$ is a \Hpath[2] in $H$ from
    $x$ to~$y$.

    We now show the backward direction of (*).
    Suppose that $H$ contains a \Hpath[2] $Q$ starting at $x$ and ending at $y$.
    As we have seen before, the neighbors of $d_2$ and $d_3$ in $Q$ can only be $d_2$ (for $d_3$),
    $d_3$ (for $d_2$) and $d_1$ and $\beta'$ (for both).
Thus, $d_2$ and $d_3$ must be consecutive in $Q$, otherwise the first one
    will be preceded and followed by $d_1$ and $\beta'$ in some order, and we
    will not be able to have neighbors for the second one.
    We also know that $d_2$ and $d_3$ are preceded by one of $d_1$ and $\beta'$,
    and followed by the other one.
Let $Q_D$ be the contiguous subsequence of~$D$ containing the four vertices
    $\{d_1, d_2, d_3, \beta'\}$.
    From our study of walks of length~$2$ in~$H$, we know that to
    go from a node-vertex to an edge-vertex, one must go through $Q_D$.
    Now, because $x$ is a node-vertex, there can only be node-vertices until
    $Q_D$, and because we cannot reach $d_1$ directly from a node-vertex, we
    know that $Q_D$ in fact must start with $\beta'$, so it must end with $d_1$.
    Similarly, because $y$ is an edge-vertex, there can only be edge-vertices between $d_1$ and $y$.
    Thus, letting $Q'$ be the prefix of~$Q$ until $\beta'$ (i.e., until
    the first vertex of~$Q_D$), we know that
    $Q'$ must visit only node-vertices, and it must visit all of them because
    afterwards $Q$ can only visit edge-vertices. Thus, $Q'$ starts with $x =
    \alpha'$, it ends with $y = \beta'$, it contains precisely all node-vertices, and adjacent
    pairs in~$Q'$ are connected by a walk of length~$2$ in~$H$ hence (from our
    study of length-$2$ walks in~$H$) the corresponding vertices in~$G$ are
    adjacent. Thus, taking the vertices of~$G$ in the order of the
    corresponding vertices in~$Q'$, we have obtained
    the desired \Hpath[1] from~$\alpha$ to~$\beta$ in~$G$.
This completes the reduction and concludes the proof.
\end{proof}

We now show the hardness of the \Hcycle[2] problem by reducing from the \Hpath[2]
with specified endpoints:

\twohcycle*

\begin{proofsketch}
    Let $G$ be an input graph and $\alpha, \beta \in V(G)$ be two different vertices.
    The reduction is to attach the gadget $D$ from \cref{fig:gadget2hcycle} to a copy of $G$
    by adding the edges $\edge{\alpha, d_0}$ and $\edge{\beta, d_1}$.
    We then prove that $G$ has a \Hpath[2] from $\alpha$ to $\beta$ if and only if the new graph $H$ has a \Hcycle[2].
    This is done by observing that the vertices $d_3, d_5, d_7$ of the gadget $D$ have very constrained neighborhoods in any \Hcycle[2] of $H$,
    which forces a specific ordering of the vertices of $D$ in such a cycle.
\end{proofsketch}

    \begin{figure}[ht]
    \centering
        \begin{tikzpicture}
            \node[circle, draw] (0) at (0,0) {$d_0$};
            \node[circle, draw] (1) at (4,0) {$d_1$};
            \node[circle, draw] (2) at (2,0) {$d_2$};
            \node[circle, draw] (3) at (2,1) {$d_3$};
            \node[circle, draw] (4) at (0,2) {$d_4$};
            \node[circle, draw] (5) at (-1,2) {$d_5$};
            \node[circle, draw] (6) at (4,2) {$d_6$};
            \node[circle, draw] (7) at (5,2) {$d_7$};

            \draw (0) -- (2);
            \draw (1) -- (2);
            \draw (2) -- (3);
            \draw (0) -- (4);
            \draw (4) -- (5);
            \draw (4) -- (6);
            \draw (1) -- (6);
            \draw (7) -- (6);
        \end{tikzpicture}
        \caption{Gadget $D$ for the proof of \cref{thm:2hcycle}, found with the help of a computer program.}\label{fig:gadget2hcycle}
    \end{figure}

\begin{toappendix}
    We now provide the full proof of the following result.

    \twohcycle*

\begin{proof}
    We reduce from the \Hpath[2] problem with specified endpoints,
    which we just proved to be NP-hard in \cref{thm:2sehpath}.
    Let $D$ be the specific gadget given in \cref{fig:gadget2hcycle}.

    Given an input graph $G$ and endpoint vertices $\alpha \neq \beta$ of~$G$,
    Let $H$ be the graph obtained by copying $G$ and connecting it to $D$
    by adding the two edges $\edge{\alpha, d_0}$ and $\edge{\beta, d_1}$.
    The construction of~$H$ from~$G$ can obviously be done in polynomial time, and it
    suffices to show that the reduction is correct by proving the following
    claim (*): $G$ has a \Hpath[2] from $\alpha$ to $\beta$ if and only if $H$
    has a \Hcycle[2].

    We first prove the forward direction of (*).
    Suppose $G$ has a \Hpath[2] $P$ starting at $\alpha$ and ending at $\beta$.
    Then one can check that $P' \coloneq P
    (d_6,d_5,d_0,d_3,d_1,d_7,d_4,d_2)$
is a \Hcycle[2] of $H$.

    We now prove the backward direction of~(*).
    Suppose $H$ has a \Hcycle[2] $C$. 
We will examine the vertices of~$C$, and
    talk of \emph{neighbors} to mean pairs of vertices that occur consecutively
    in~$C$ (including possibly as the first and last vertices)-- note that
    these vertices are then connected by a walk of length~$2$ but they are not
    necessarily neighbors in~$H$.
    We observe that the only vertices connected to $d_3$ by a walk of length~$2$
    are $d_0$ and $d_1$. Thus, these must be its two neighbors in $C$.
    For similar reasons,
    $d_5$ must have $d_0$ and $d_6$ as its neighbors in $C$ and $d_7$ must have $d_1$ and $d_4$ as its neighbors in $C$.
    Because every vertex must have exactly $2$ neighbors in $C$, 
    we deduce that $d_0$ must be next to both $d_3$ and $d_5$ while $d_1$ must be next to both $d_3$ and $d_7$.
    This means that, up to reversing and/or conjugating the cycle, we know that
    the sequence of vertices $S = (d_6,d_5,d_0,d_3,d_1,d_7,d_4)$ must occur
    consecutively in $C$.
    Furthermore, $d_2$ must be right before or right after this subsequence. If it were not, the previous and next vertices
    of the sequence would be $\alpha$ and $\beta$ and we would either go on $d_2$ and be stuck on it, or go to different vertices and have no way to ever come back to $d_2$.
    Assume that $d_2$ is right before $S$, the other case is analogous.
    Then the only vertex connected to~$d_4$ by a walk of length~$2$ that is not
    in the sequence is $\alpha$, and we deduce that $\alpha$ is the vertex right after $S$.
    The same holds for $d_2$ and $\beta$, thus $\beta$ is the vertex before $d_2$.
    By removing every vertex of $D$ we are then left with a \Hpath[2] of $H$
    that starts at $\alpha$ and ends at $\beta$, 
    which is also a \Hpath[2] of $G$ that starts at $\alpha$ and ends at $\beta$.
    This is what we wanted to obtain, so the backward direction of (*) is
    established, which concludes the proof.
\end{proof}
\end{toappendix}

 \subsection{The Case \texorpdfstring{$\bm{S = \{2,4\}}$}{S = \{2,4\}}}\label{sec:24}
\begin{toappendix}
  \subsection{Proofs for \texorpdfstring{$\bm{S = \{2,4\}}$}{S = \{2,4\}}}
  \label{apx:24}
\end{toappendix}

We now turn to the case of the \Hcycle[2,4] problem. We show NP-hardness via a
Cook reduction from the \Hcycle[1,2] problem; formally:

\twofourhcycle*

Our hardness proof uses Cook reductions for their ability to reduce an instance
to multiple instances with multiple oracle calls.
While we suspect that \Hcycle[2,4] should also be NP-hard under the more
standard notion of Karp reductions, we have not been able to show this.

\cref{thm:24hcycle} will directly follow from the following result.

\begin{propositionrep}\label{prop:24hcycleconstruction}
  Given an input graph $G$ and two vertices $x \neq y$ at distance at most~$2$ in~$G$,
  we can construct in polynomial time a graph $H_{x,y}$
    such that:
    \begin{itemize}
        \item if $G$ has a \Hcycle[1,2] where $x$ and $y$ are consecutive,
          then $H_{x,y}$ has a \Hcycle[2,4],
        \item if $H_{x,y}$ has a \Hcycle[2,4], then $G$ has a \Hcycle[1,2].
    \end{itemize}
\end{propositionrep}

\begin{proof}
    Let $G$ be a graph, and $x \neq y$ two vertices of $G$ at distance at most~$2$.
    Let $z$ be a neighbor of~$x$ towards $y$ (defined formally below).
    Let us now build the graph $H_{x,y}$ by taking the incidence graph of $G$
    and connect each of $x$ and $z$ respectively to the vertices~$d_1$ and~\(d_2\) of the gadget $D$
    shown in \cref{fig:gadget24hcycle}.

    Formally, we construct $H_{x,y}$ as follows:
    \begin{itemize}
        \item Build the incidence graph $I(G)$ by doing the following:
        \begin{itemize}
        \item For each vertex $v \in V(G)$, add a corresponding vertex $v' \in
          I(G)$. We call these \emph{node-vertices}.
        \item For each edge $\edge{u,v} \in E(G)$, add a new vertex $w_{uv}$ to
          $I(G)$ and connect it to both $u'$ and $v'$.
          We call these \emph{edge-vertices}.
        \end{itemize}
        \item Let $H_{x,y}$ be $I(G)$ together with a copy of the gadget $D$
        from \cref{fig:gadget24hcycle}.
        Define $z$ this way: if $x$ and $y$ are adjacent then $z = y$,
        else let $z$ be any common neighbor of $x$ and $y$.
        Now make $H_{x,y}$ connected by adding two new edges
        \edge{x', d_1} and \edge{z', d_2}.
    \end{itemize}

    We now prove that this construction satisfies the two required properties,
    starting with the first item of the proposition statement.
    Let $C$ be a \Hcycle[1,2] of~$G$ where $x$ and $y$ are consecutive, 
    and let us show that $H_{x,y}$ has a \Hcycle[2,4].
    Up to conjugating and reversing~$C$, we can suppose without loss of generality that $C$
    starts at $x$ and that $y$ immediately follows,
    so let us write $C = (x, y, v_1, v_2, \dots, v_k)$.
    Let us build the \Hcycle[2,4] of~$H_{x,y}$ by initializing $C'$ with $x'$
    and adding enough vertices of $D$ to be able to reach an edge-vertex,
    namely: $C' = (x', d_5, d_4)$.
    Here, the intermediate walks are $x' - d_1 - d_3 - d_4 - d_5$ and $d_5
    - d_7 - d_6 - d_5 - d_4$, both of which have length $4$.

    Now, consider the graph $G'$ that is a copy of $G$ where we add two vertices $d_3'$ and $d_5'$ and the three edges \edge{x, d_3'}, \edge{z, d_3'}, and \edge{d_3', d_5'}.
    Observe that the edges of~$G'$ precisely correspond to the edge-vertices
    of~$H_{x,y}$ plus the three vertices $d_1$, $d_2$, and $d_4$.
    By construction $G'$ has at least $3$ edges,
    and further the edges \edge{d_3', d_5'} and \edge{z, d_3'} are connected by a walk of
    length~$2$ in the line graph~$L(G')$ of~$G'$, as witnessed by the walk 
    $\edge{d_3', d_5'} - \edge{x, d_3'} - \edge{z, d_3'}$.
    By \cref{lem:12pathlinetree}, there is a \Hcycle[1,2] $C''$ in $L(G')$ where \edge{d_3', d_5'} and \edge{z, d_3'} are consecutive.
    We can transform~$C''$ into a simple $\{2,4\}$-path $P''$ of $H_{x,y}$ that
    starts at $d_4$, visits all the edge-vertices of $H_{x,y}$ and $d_1$,
    and ends on $d_2$, by replacing each vertex of~$L(G')$ by the corresponding edge-vertex
    in~$H_{x,y}$ and replacing \edge{d_3', d_5'},  \edge{x, d_3'}, and \edge{z, d_3'},
    by $d_4$, $d_1$, and $d_2$, respectively.
    We append $P''$ to~$C'$ (without the first vertex~$d_4$ as it is already
    in~$C'$).
    Then from there, we go to $d_6$ through $d_2 - d_3 - d_4 - d_5 - d_6$, then we go to $d_7$ through $d_6 - d_5 - d_7$,
    and finally we go to $d_3$ through $d_7 - d_6 - d_5 - d_4 - d_3$.

    We can now complete the construction of $C'$ by concatenating our current $C'$
    with $C_{end} = (y', v_1', v_2', \dots, v_k')$,
    where we go to $y'$ through $d_3 - d_2 - y'$ if we took $z = y$
    (a walk of length~$2$),
    or otherwise through $d_3 - d_2 - z' - w_{zy} - y'$ (a walk of
    length~$4$).
    Then, all consecutive vertices in $C_{end}$ are connected by a walk of
    length $2$ or $4$ because they were connected by a walk of length $1$ or $2$ in $G$ by
    definition of~$C$ being a \Hcycle[1,2] of~$G$.
    Last, there is a walk of length $2$ or $4$ from $v_k'$ to the first vertex
    $x'$ of~$C'$ for the same reason.
    Thus, $C'$ is indeed a \Hcycle[2,4] of $H_{x,y}$.

    \medskip

    We now show the second item of the proposition statement:
    let $C'$ be a \Hcycle[2,4] of~$H_{x,y}$,
    and let us prove that $G$ has a \Hcycle[1,2].
    Up to conjugating~$C'$, we can assume without loss of generality that $C'$ starts at a node-vertex $v$.
    Because $C'$ must at some point visit an edge-vertex and come back to $v$,
    it means that $C'$ must switch between node-vertices and edge-vertices at
    least twice (note that if $C'$ starts with~$v$ and ends with an edge-vertex
    then this is also counted as a switch, namely, a switch between the last
    vertex and the first vertex $v$ as we close the cycle).
    Now, if we remove the vertices of $D$ from $H_{x,y}$, we are left with a
    graph that is bipartite between node-vertices and edge-vertices. Thus,
    as long as no vertex of $D$ is visited (in~$C'$ or as an intermediate vertex
    of a walk in~$C'$), and because we are only allowed walks of even length, we
    know that $C'$ must visit either only node-vertices or only edge-vertices
    between visits to vertices of~$D$ (as vertices of~$C'$ or as intermediate
    vertices). 
    More precisely, because the only way to switch between node-vertices and
    edge-vertices is to traverse the triangle $\{d_5,d_6,d_7\}$ (without which
    $H_{x,y}$ would be bipartite), and because the maximal walk length $4$ is less than the distance
    between the node-vertices and the vertices of this triangle, we even know
    that $C'$ can only change between node-vertices and edge-vertices by
    a visit to at least one vertex of~$D$ which occurs in~$C'$ (i.e., not just
    as an intermediate vertex). We will now show that we can isolate in $C'$ a
    contiguous sequence $N$ which visits all node-vertices.

    To justify this claim, it is easy to see by construction that $D$ only has two vertices
    which may occur in~$C'$ consecutively to a node-vertex,
    namely, the vertices $d_3$ and
    $d_5$, which are the only vertices of~$D$ connected by a walk of length~$2$
    or~$4$ to a node-vertex. 
    Thus, consider the first time that $C'$ visits a vertex of~$D$. As $C'$ was
    visiting node-vertices until then, the first vertex of~$D$ that occurs
    in~$C'$ is $d_3$ or~$d_5$. We argue that the other vertex of $d_3, d_5$
    cannot occur before $C'$ exits~$D$. Indeed, assuming by contradiction that
    it does, there are three possibilities. The first case is that the first
    vertex of~$C'$ after vertices of~$D$ is a node-vertex, and then we can never
    re-enter $D$ and so we can never visit the edge-vertices, a contradiction.
    The second case is that the first vertex of~$C'$ after vertices of~$D$ is an
    edge-vertex, but then we need to re-enter $D$ to switch back to
    node-vertices and we will have no way to exit~$D$ to a node-vertex because
    both $d_3$ and $d_5$ are taken. The third case is that $C'$ ends after this
    visit to~$D$, but this is impossible because the edge-vertices were not
    visited. Hence, in all three cases we know that $C'$ exits $D$ without
    visiting the other vertex from~$d_3,d_5$. In particular, $C'$ exits $D$ to
    an edge-vertex.
    We then know that the other vertex from~$d_3,d_5$ must be used to exit $D$
    when $C'$ will switch back to the node-vertices. This reasoning implies that
    $C'$ cannot switch between node-vertices and edge-vertices more than two times.
    In conclusion, up to conjugating $C'$, we can write it as $C' = N D_1 E D_2$
    where $N$ consists precisely of all the node-vertices, where
    $D_1$ and $D_2$ are only composed of vertices from $D$, and where $E$ 
    contains of all the edge-vertices (note that we have not ruled out that
    vertices from $D$ may appear in~$E$, but this does not matter).

    We now focus our attention on~$N$ and show that it yields a \Hcycle[1,2]
    in~$G$. To see why, observe that
    consecutive node-vertices in $C$ must correspond to walks of length $2$
    or $4$ in $H_{x,y}$. These
    correspond to walks of length $1$ or $2$ in $G$,
    or to walks that go via $D$ but the only such walk between node-vertices
    in~$H_{x,y}$ is the walk of length~$4$ connecting $x'$ and $z'$ for which the
    corresponding vertices $x$ and $z$ were adjacent in~$G$.

    Thus, if we take the original vertices of $G$ corresponding to the node-vertices
    in the order they appear in $N$, we get a \Hpath[1,2] in $G$, from the
    original vertex $a$ of the first
    vertex right after $D_2$ in~$C'$ (i.e., the first vertex of $N$), denoted~$a'$, to
    the original vertex $b$ of the last vertex right before $D_1$ in~$C'$, denoted~$b'$.
    To complete the proof, we need to show that this path is a cycle, i.e., that 
    $a$ and $b$ are at connected by a walk of length $1$ or $2$ in $G$.
    Let us analyze which possible node-vertices $a'$ and $b'$ can be.
    There are two types of node-vertices that are connected to~$D$ by walks of
    length $2$ or $4$:
    \begin{itemize}
        \item $x'$ and $z'$ can reach either $d_3$ or $d_5$.
        \item Any node-vertex different from $x'$ or $z'$ whose original in~$G$ is adjacent
              to $x$ or $z$ can reach $D$ only on $d_3$.
    \end{itemize}
    Recall now from our previous reasoning that $D_1$ must start with one of
    $d_3,d_5$ and $D_2$ must end with the other vertex among $d_3,d_5$. Thus,
    one of $a'$ and $b'$ occurs consecutively to~$d_5$, so one of $a'$ and $b'$ must
    be $x'$ or $z'$.
    The other vertex must then be either the other one of $x'$ and $z'$
    or a node-vertex whose original in~$G$ is adjacent to
    $x$ or $z$. Recall now that $z$ was picked to be adjacent to~$x$ in~$G$, so
    in all the possibilities considered we know that $a$ and $b$ in~$G$ must be
    at a distance of at most~$2$.

    Thus $G$ contains a \Hcycle[1,2], which concludes the proof.
\end{proof}

Note that, when $H_{x,y}$ has a \Hcycle[2,4], then the \Hcycle[1,2] in~$G$ is
not necessarily one where $x$ and $y$ are consecutive. Nevertheless, the result
of \cref{prop:24hcycleconstruction} suffices to design a Cook reduction that
shows the NP-hardness of the \Hcycle[2,4] problem:

\begin{proof}[Proof of \cref{thm:24hcycle}]
  We do a Cook reduction from the \Hcycle[1,2] problem, which as we already mentioned in \cref{prop:1hcycle} is
NP-hard. Given an input $G$ to the \Hcycle[1,2] problem, assuming without loss
of generality that it has at least two vertices,
we consider every pair $x,y$ of
vertices at distance at most~$2$ in $G$ 
and construct the corresponding graph $H_{x,y}$: this amounts to polynomially
many graphs which are all built in polynomial time. Now, if $G$ has a
\Hcycle[1,2] $C$ then for any choice of contiguous vertices $x,y$ of~$C$, we
know that $H_{x,y}$ has a \Hcycle[2,4]. Conversely, if some $H_{x,y}$ has a
\Hcycle[2,4] then $G$ has a \Hcycle[1,2]. This implies that $G$ has a
\Hcycle[1,2] if and only if some $H_{x,y}$ has a \Hcycle[2,4], so we can indeed
solve \Hcycle[1,2] in polynomial time with access to an oracle for
\Hcycle[2,4]. This establishes the Cook reduction and concludes the proof.
\end{proof}

Thus, all that is left to show \cref{thm:24hcycle} is to
prove \cref{prop:24hcycleconstruction}. To do so, we will need the following
intermediate lemma which is a stronger version of~\cite[Theorem 2]{On_the_line_gra_Nebesk_1973}.
This lemma originally proves that the square of a line graph is always Hamiltonian, i.e., that the line graph of any (connected) graph always has a \Hcycle[1,2],
and we need a slightly stronger version that allows us to choose two vertices of the line graph to be consecutive in the \Hcycle[1,2].
\begin{toappendix}
We now give the full proof of the following lemma, which adapts the proof of \cite[Theorem 2]{On_the_line_gra_Nebesk_1973}.
\end{toappendix}
\begin{lemmarep}[Strengthening of Theorem~2 of \cite{On_the_line_gra_Nebesk_1973}]
  \label{lem:12pathlinetree}
    Let $G$ be a graph with at least $3$ edges.
    Let $L(G)$ be the line graph of~$G$. Let $\alpha$ and $\beta$ be
    two vertices of $L(G)$ connected by a walk of length $2$ 
    in $L(G)$. Then there exists a \Hcycle[1,2] in $L(G)$ where $\alpha$ and
    $\beta$ are consecutive. 
\end{lemmarep}

\begin{proof}
    Let $G$ be a graph with at least $3$ edges and $L(G)$ be its line graph.
    Let $\alpha$ and $\beta$ be two vertices of $L(G)$ connected by a walk of length $2$ in $L(G)$.
    We first show that it is sufficient to prove the claim for trees:
    Consider a spanning tree $T_1$ of $G$.
    Color the edges of $T_1$ in blue in $G$.
    Subdivide each uncolored edge of $G$ by adding a new vertex in the middle of
    the edge, and color one of the two new edges in red and the other in blue.
    Let $T_2$ be the graph composed of all the blue edges.
    Obviously, $T_2$ is a tree with at least $3$ edges.
    It is easy to see that $L(T_2)$ is isomorphic to a spanning subgraph of $L(G)$.
    This means that we just need to prove that for any tree \(T\) with at least
    three edges, and for any vertex \(a\) and \(b\) of \(L(T)\) such that
    \(a\) and \(b\) are at walk distance \(2\) in \(L(T)\),
    there exists a \Hcycle[1,2] in \(L(T)\) where \(a\) and \(b\) are consecutive.
    Indeed, this would imply that $L(G)$ also has such a \Hcycle[1,2] simply by
    taking the \Hcycle[1,2] of $L(T_2)$ and following the same sequence of vertices.
    Thus, we can now focus on proving the claim for trees.

    Let $T$ be a tree with at least $3$ edges.
    Let $L(T)$ be the line graph of~$T$.
    Let $\alpha$ and $\beta$ be two vertices of $L(T)$ such that there exists a walk of length $2$ between them in $L(T)$.
    We will show that there exists a \Hcycle[1,2] in $L(T)$ where $\alpha$ and
    $\beta$ are consecutive.
    We prove the claim by induction on the number of edges of $T$. 
    The base case is when $T$ has $3$ edges,
    and then there are two possibilities. First, if $T$ is a path with 4 vertices, then $L(T)$
    is a path with 3 vertices and the claim is immediate. Second, if $T$ is a
    star with 3 branches, then $L(T)$ is the complete graph on 3 vertices, and
    the claim is also immediate.
    For the induction case, fix $n > 3$, suppose the claim holds for any tree
    with strictly fewer than $n$ edges, and let $T$ be a tree with $n$ edges. We
    will show that the claim is true on~$T$.
    Let $L(T)$ be the line graph of~$T$ and let $\alpha$ and $\beta$ be vertices of $L(T)$ such that there exists a walk of length $2$ between them in $L(T)$.

    There are two cases: either $T$ is a path, or $T$ is not a path.
    First, if $T$ is a path, then $L(T)$ is the path with $n-1$ vertices $(v_1, \dots, v_{n-1})$, and
    we can take the \Hcycle[1,2] $C = (v_1, v_3, \dots, v_{n-1}, v_{n-2}, \dots,
    v_2)$ if $n-1$ is odd, and the same with $v_{n-1}$ and $v_{n-2}$ swapped if
    $n$ is even. This choice of~$C$ ensures in fact that every pair of vertices at distance
    $2$ in $L(T)$ are consecutive in $C$, hence this is true in particular of~$\alpha$ and
    $\beta$.
    
    Second, if $T$ is not a path, then $T$ has at least $3$ leaves.
    We connect these leaves to a vertex of degree $>2$ by iteratively
    taking the unique neighbor of the current vertex until we reach such a vertex.
    (The vertices of degree $>2$ that we reach in this fashion may or may
    not be distinct.)
    The three paths constructed in this way do not share any edges, so let us pick one that contains
    neither the edge corresponding to $\alpha$ nor the edge corresponding to $\beta$.
    Formally, letting $v_k$ be the chosen leaf and $v_0$ the vertex of degree at
    least~$3$ that we reach, 
    we obtain by this argument a simple path $v_0, \ldots, v_k$ in~$T$
    where $v_0$ has degree $>2$,
    all vertices $v_i$ with $0 < i < k$ have degree~2,
    $v_k$ has degree~$1$,
    and the edges of that path are all different from $\alpha$ and $\beta$.

    Let us consider the tree $T_0$ obtained from $T$ by deleting the vertices
    $v_1,\dots, v_k$ and all edges incident to them, which we then root at~$v_0$.
    By $u_1,\dots, u_i$ we denote the children of $v_0$ in
    $T_0$; as $v_0$ had degree $>2$ in~$T$, we know that $i \geq 2$.
    Now, because $T_0$ has at most $n-1$ edges, and $L(T_0)$ contains $\alpha$ and $\beta$,
    we can use the induction hypothesis on~$T_0$ to
    obtain a \Hcycle[1,2] $C$ in $L(T_0)$ where $\alpha$ and $\beta$ are
    consecutive. We will now modify $C$ to be a \Hcycle[1,2] of~$L(T)$ where $\alpha$
    and $\beta$ are still consecutive, by inserting back the edges of the path
    $v_0, \ldots, v_k$ which we had removed.

    Let $x$ and $y$ be two vertices of $L(T_0)$ such that $x$ and $y$ are consecutive in $C$ (also possibly as first and
    last vertices),
    such that $y$ is incident to $v_0$, i.e. $y = \edge{v_0, u_y}$ for some child $u_y$ of $v_0$ in $T_0$,
    and such that the edge of $T_0$ corresponding to $x$ is incident to one of the $u_i$ in $T_0$.
    Such a pair $(x,y)$ must exist because we will at some point need to go from the subtree of $T_0$
    rooted at $u_i$ to the subtree rooted at $u_j$ for some $i \neq
    j$. (Note that here we use the fact that $v_0$ has degree $\geq 2$ in~$T_0$.)
    There are three ways to achieve this with two consecutive vertices of
    $L(T_0)$ at distance $\leq 2$:
    either by going from \edge{v_0, u_i} to \edge{v_0, u_j}, 
    or by going from \edge{u_i, w_i} to \edge{v_0, u_j} 
    (where $w_i$ is a neighbor of~$u_i$ which is different from~$v_0$)
    or by going from \edge{v_0, u_i} to \edge{u_j, w_j}
    (where $w_j$ is a neighbor of~$u_j$ which is different from~$v_0$).
    In any case, we can pick $x$ and $y$ accordingly.

    Now, we want to insert the missing vertices between $x$ and $y$. However, if $\{ x,y \} = \{ \alpha, \beta \}$,
    this will not work, as we want to keep them consecutive.
    To this end, we will show how to modify $x$ and $y$ to ensure that at most one of them is in $\{\alpha, \beta\}$,
    such that they still respect the properties listed above.
    If $\{ x,y \} = \{ \alpha, \beta \}$, the other vertex consecutive to $y$ in $C$ is either corresponding to an edge incident to a $u_l$
    (which is forced if $u_y$ is a leaf in $T_0$),
    or it is deeper in subtree rooted at $u_y$.
    In the first case, this other neighbor of $y$ can be
    called $x$ in the remainder of the proof.
    In the other case, the only way to exit the subtree rooted at $u_y$ is to go from an edge incident to $u_y$ that is not $y$ towards an edge incident to $v_0$,
    and they can become the new $x$ and $y$ respectively.
    
    We now have $x$ and $y$ such that at most one of them is in~$\{\alpha,
    \beta\}$, such that $y$ is incident to $v_0$,
    and such that $x$ is incident to a $u_i$.
    By $P$ we denote the $\set{1,2}$-path in $L(T)$ such that if $k = 1$, then
    $P = (x, \edge{v_0,v_1}, y)$,
    and if $k \geq 2$, then $P = (x, \edge{v_0,v_1}, \edge{v_2,v_3},
    \dots, \allowbreak \edge{v_{g-3},v_{g-2}}, \allowbreak \edge{v_{g-1},v_g},
    \allowbreak \edge{v_h, v_{h-1}}, \dots, \edge{v_2,v_1}, y)$,
    where $g$ is the greatest odd integer not exceeding $k$ and $h$ is the greatest even integer not exceeding $k$.
    If in $C$ we replace \edge{x, y} by $P$, we obtain a \Hcycle[1,2] in $L(T)$ and because we insert nothing between $\alpha$ and $\beta$, they are still consecutive in this resulting cycle.
    Hence, we have shown that there is a \Hcycle[1,2] in
    $L(T)$ where $\alpha$ and $\beta$ are consecutive, which concludes the proof
  of the induction hypothesis. Thus, we have shown the desired claim by
induction, which concludes the proof.
\end{proof}

With this lemma in hand, we are now ready to prove \cref{prop:24hcycleconstruction}.

\begin{figure}
    \centering
    \begin{tikzpicture}
        \node[circle, draw] (d1) at (-1.44,1) {$d_1$};
        \node[circle, draw] (d2) at (-1.44,-1) {$d_2$};
        \node[circle, draw] (d3) at (0,0) {$d_3$};
        \node[circle, draw] (d4) at (2,0) {$d_4$};
        \node[circle, draw] (d5) at (4,0) {$d_5$};
        \node[circle, draw] (d6) at (6,0) {$d_6$};
        \node[circle, draw] (d7) at (5,1.44) {$d_7$};

        \draw (d1) -- (d3);
        \draw (d2) -- (d3);
        \draw (d3) -- (d4);
        \draw (d4) -- (d5);
        \draw (d6) -- (d5);
        \draw (d7) -- (d6);
        \draw (d7) -- (d5);
    \end{tikzpicture}
    \caption{Gadget $D$ for the proof of \cref{prop:24hcycleconstruction}}\label{fig:gadget24hcycle}
\end{figure}

\begin{proof}[Proof sketch of \cref{prop:24hcycleconstruction}]
  Given an input graph $G$ on which we want to determine the existence of a
  \Hcycle[1,2] and two of its vertices \(x\) and \(y\), we build $H_{x,y}$ by taking the incidence graph of~$G$ (formed
  of node-vertices and edge-vertices) and
  attaching the gadget from \cref{fig:gadget24hcycle} with edges
  from \(d_1\) and \(d_2\) respectively to the vertices
  $x'$ and $z'$ corresponding to~$x$ and to a vertex $z$ which is a neighbor
  of~$x$ towards $y$ (specifically, we take $z=y$ if $x$ and $y$ are at distance $1$, and we take $z$ to be a common neighbor of~$x$ and~$y$ if they are at distance $2$).

  For the forward direction, we show that a \Hcycle[1,2] $C$ of~$G$ where $x$ and
  $y$ are adjacent can be lifted to a \Hcycle[2,4] $C'$ of~$H_{x,y}$. At a high
  level, $C$ shows us how to visit the node-vertices of~$H_{x,y}$,
  and we use $D$ to toggle to the edge-vertices and back. Further, we use
  \cref{lem:12pathlinetree} to build the portion of the cycle that visits the
  edge-vertices of~$H_{x,y}$ while starting and ending at appropriate vertices.

  For the backward direction, given a \Hcycle[2,4] $C'$ of~$H_{x,y}$, we use the fact
  that $H_{x,y}$ is almost bipartite except for the triangle $\{d_5,d_6,d_7\}$
  in~$D$. This allows us to show that $C'$ can only switch one time from
  node-vertices to edge-vertices and one time back, so that $C'$ must contain a
  contiguous \Hpath[2,4] $N$ that visits all the node-vertices, which is preceded
  and succeeded by visits to~$D$. Further
  considerations on the structure of~$H_{x,y}$ ensure that $N$ must in fact be 
  a \Hcycle[1,2] of~$G$, i.e., its starting and ending vertices 
  correspond to vertices at a distance of at most $2$ in~$G$.
\end{proof}

\section{Trivial Case: \safebm{$S = \{1,2,4\}$}}\label{sec:trivial}
In this section we show that every graph admits a \Hcycle[1,2,4] which we can construct in linear time.

\onetwofourhcycle*

Up to taking a spanning tree, this is equivalent to showing that every tree
admits a \Hcycle[1,2,4], which we do in the following lemma (with an additional property that is needed for its own proof).

\begin{lemmarep}\label{lem:124hcycletree}
Let $T$ be a tree rooted at $r$ such that $|V(T)| \geq 2$. Then there exists a \Hcycle[1,2,4] 
$(r,v_1,\dots,v_k)$ of~$T$ such that $v_1$ is a child of~$r$.
Furthermore, such a cycle can be computed in linear time.
\end{lemmarep}

\begin{proofsketch}
    The idea of the proof is to do an induction on the number of vertices of $T$. 
    We distinguish the children of~$r$ in~$T$ between those children that are
    leaves and those children that are roots of non-trivial subtrees.
    If there are leaf children, then we can simply visit them in order with a
    1-hop from $r$ and 2-hops between each of them, easily ensuring the
    requirement that the second vertex of the cycle is a child of~$r$.
    For each non-leaf child $v_i$ of~$r$,
    we use the induction hypothesis to get a cycle $C_i$ which visits the
    subtree of~$T$ rooted at~$v_i$ while starting at~$v_i$ and visiting a child of~$v_i$
    immediately afterwards. Up to conjugating these cycles and reversing 
    every
    other cycle, we can
    stitch these cycles together to visit the subtrees rooted at non-leaf
    children of~$r$:
    we start at the first such vertex~$v_1$
    (reached via a 1-hop from~$r$, or a 2-hop
    from the last leaf child of~$r$), follow~$C_1$ (after conjugation and
    reversal), and finish on a child $w_1$ of~$v_1$,
    then we do 4-hop to a child $w_2$ of~$v_2$ and do $C_2$ in reverse order 
    to finish on~$v_2$. At the end of this process, we finish either on the last
    non-leaf child of~$r$ or a child thereof, and we can close the cycle and go
    back to~$r$ with a 1-hop or a 2-hop.

    This inductive proof can be transformed into a linear-time
    algorithm that constructs the desired \Hcycle[1,2,4] of~$T$.
\end{proofsketch}

\begin{proof}
We proceed by induction on $|V(T)|$.

\textbf{Base case.}
The base case is when $|V(T)|=2$, in which case $T$ is the tree with one edge and the claim is trivial.

\textbf{Induction step.}
Let $n \in \NN, n \ge 3$, and assume the statement holds for all trees with at most $n-1$ vertices.  
Let $T$ be a tree rooted at $r$ with $n$ vertices.

Partition the children of $r$ into leaf children $u_1,\dots,u_\ell$ and non-leaf children 
$v_1,\dots,v_m$, where \(\ell,m\ge 0\) and \(\ell+m \ge 1\)
because \(n \geq 3\). (In fact we even know that \(\ell+m \ge 2\), but we do not
need it.)
For each $i$, denote by $T_i$ the subtree of~$T$ rooted at $v_i$, i.e., the connected component of
$T\setminus\{r\}$ that contains $v_i$. Note that, as $v_i$ is not a leaf, we
must have $|V(T_i)|\ge 2$ for all $i$.

Because every $T_i$ has less than $n$ vertices, we know by induction hypothesis
that each $T_i$ admits a \Hcycle[1,2,4] $C_i$ starting at
$v_i$ and starting with a 1-hop to a child of $v_i$ in $T_i$, denoted $w_i$.

We will now construct the \Hcycle[1,2,4] of $T$ in two stages. First, we will visit all leaf children 
of $r$ (if some exist), and second we will visit all vertices in each subtree
$T_i$ one after the other (if some exist).

\textbf{Stage 1: The leaf children.}  
If $\ell>0$, let us traverse the leaves in order: start from $r$ and move to $u_1$ (walk of length $1$).
Then for $j = 1,\dots,\ell-1$ move from $u_j$ to $u_{j+1}$ via $u_j-r-u_{j+1}$ (walk of length $2$).
If $m=0$, we close the cycle with $u_\ell-r$ (walk of length $1$) and we
obtain a
\Hcycle[1,2,4] of $T$ which starts at a child of~$r$, so we are finished.
Otherwise, if $m>0$, continue to stage~2.

\textbf{Stage 2: All the $T_i$.}
First, go to $v_1$. If $\ell>0$, this is done via $u_\ell - r - v_1$ (walk of length $2$);
if $\ell=0$, this is done via $r - v_1$ (walk of length $1$).
Then traverse the \Hcycle[1,2,4] $C_1$ of $T_1$ after conjugation and reversal,
going from $v_1$ to $w_1$ while visiting all vertices of~$T_1$.
Next, traverse the other non-leaf children trees $T_2,\dots,T_m$ as follows:
\begin{enumerate}
    \item For $i = 2, \dots, m$, we alternate those two operations:
    \begin{enumerate}
        \item If $i$ is even, then we go to $w_i$ from $w_{i-1}$ through
          $w_{i-1} - v_{i-1} - r - v_i - w_i$, which is a walk of length~$4$.
          We traverse the cycle $C_i$ of
          $T_{i}$ in order (after conjugation) to go from $w_{i}$ to $v_{i}$
          while visiting all vertices of~$T_i$.
        \item If $i$ is odd, then we go to $v_i$ from $v_{i-1}$ through
          $v_{i-1} - r - v_i$, which is a walk of length~$2$.
          We traverse the cycle $C_i$ of $T_i$ (after conjugation) in reverse
          order, starting from $v_i$ and ending at $w_i$, while visiting all
          vertices of~$T_i$.
    \end{enumerate}
    This alternation continues until all subtrees have been visited.
    \smallskip
    \item Finally, after processing $T_{m}$, we return to $r$ from either $w_m$ or $v_m$ depending on the parity of $m$.
    If $m$ is odd, this is a walk of length $2$ through $w_m - v_m - r$;
    if $m$ is even, this is a walk of length $1$.
\end{enumerate}
The constructed sequence therefore is a \Hcycle[1,2,4] of $T$ starting at $r$
where the second vertex is a child of $r$, which establishes the inductive claim.

This construction can be implemented by a recursive algorithm which considers each vertex exactly once,
thus runs in linear time, which concludes the proof.
\end{proof}

\begin{toappendix}

We also provide the proof of the \cref{thm:124hcycle}, using the previous lemma.
\begin{proof}[Proof of \cref{thm:124hcycle}]
Let $G$ be a connected graph. If $|V(G)|=1$, then $G$ is a single vertex, say $v$, and $(v)$ is a \Hcycle[1,2,4] of $G$.
Otherwise, choose a spanning tree $T$ of $G$.
Using \cref{lem:124hcycletree}, we know that $T$ admits a \Hcycle[1,2,4] $(r,v_1,\dots,v_k)$ where $r$ is the root of $T$ and $v_1$ is a child of $r$ in $T$.
Because $T$ is a subgraph of $G$, $(r,v_1,\dots,v_k)$ is also a \Hcycle[1,2,4] of $G$.
\end{proof}
\end{toappendix}

\section{Polynomial Case: \safebm{$S = \{2,4,6\}$}}\label{sec:poly}
In this section, we show that the \Hcycle[2,4,6] problem can be solved in linear
time because the graphs with a \Hcycle[2,4,6] are precisely the non-bipartite
graphs:

\twofoursixhcycle*

To prove this theorem, we will first show that every non-bipartite graph admits a \Hcycle[2,4,6].

\begin{lemmarep}\label{lem:246hcyclenb}
    Given a non-bipartite graph $G$, we can construct a \Hcycle[2,4,6] of $G$ in linear time.
\end{lemmarep}

\begin{proofsketch}
    Let $G$ be a non-bipartite graph.
    Then $G$ must contain a simple cycle of odd length.
    Let $C = (v_0, v_1, \dots, v_{2k})$ be a shortest odd cycle in $G$.
    We construct a spanning forest rooted at each $v_i$ by performing a simultaneous breadth-first search starting from all~$v_i$.
    This way, we obtain a collection of trees $T_0, T_1, \dots, T_{2k}$, each rooted at their respective~$v_i$, and covering all vertices of $G$.
    We can use \cref{prop:123hcycle} about the triviality of the \Hpath[1,2,3] problem to create, in each $T_i$, a \set{1,2,3}-path of the tree obtained from taking only vertices at odd depth (resp., at even depth) in $T_i$,
    which translates to a \set{2,4,6}-path of those vertices in $T_i$.
    We traverse the cycle $C$ twice.
    During the first traversal, we visit the vertices $v_i$ with even index $i$, and during the second traversal, we visit the vertices $v_i$ with odd index $i$.
    When visiting a vertex $v_i$, we also enumerate all vertices of $T_i$ at even depth, and all vertices of $T_{i+1}$ at odd depth. (The indices are taken modulo $2k+1$.)
    This yields a \Hcycle[2,4,6] of $G$.
\end{proofsketch}

\newcommand{\Vie}[1]{\ensuremath{V_{#1}^{\mathrm{even}}}}
\newcommand{\Vio}[1]{\ensuremath{V_{#1}^{\mathrm{odd}}}}
\newcommand{\Pie}[1]{\ensuremath{P_{#1}^{\mathrm{even}}}}
\newcommand{\Pio}[1]{\ensuremath{P_{#1}^{\mathrm{odd}}}}

\begin{proof}
Let $G = (V,E)$ be a non-bipartite graph. Then $G$ must contain a cycle of odd length. 
Let $C = (v_0, v_1, \dots, v_{2k})$
be a shortest odd cycle in $G$, where $k \ge 1$.

\medskip

\textbf{Step 1: Constructing a spanning forest rooted at the cycle.}  
Consider the graph $G' := G \setminus C$ obtained from $G$ by removing all edges between vertices of $C$.
We construct a spanning forest $\{T_0, T_1, \dots, T_{2k}\}$ rooted at the cycle vertices as follows:

\begin{enumerate}
    \item Initialize each tree $T_i := \{v_i\}$.
    \item Perform a simultaneous breadth-first search starting from all $v_i$.  
    \item Whenever a vertex $u \in V \setminus C$ is reached for the first time from some root $v_i$, assign $u$ to $T_i$.  
    \item Continue until every vertex in $V$ is assigned to exactly one $T_i$.  
\end{enumerate}

By construction, each $T_i$ is connected, the trees are disjoint, and their union covers all vertices of $G$, forming a spanning forest rooted at the cycle vertices.

\medskip

\textbf{Step 2: Constructing Hamiltonian paths within each tree.}  
For $0 \leq i \leq 2k$ and $u$ a vertex of $T_i$, let $d_{T_i}(u)$ be the depth of $u$ in the tree $T_i$.  
Let $\Vie{i} := \{u \in T_i \mid d_{T_i}(u) \bmod 2 = 0\}$ be the set of vertices at even depth in $T_i$,
and let $\Vio{i} := \{u \in T_i \mid d_{T_i}(u) \bmod 2 = 1\}$ be the set of vertices at odd depth in $T_i$.  
We claim that we can construct a path in every $T_i$ that starts at any vertex of $\Vie{i}$ (resp. $\Vio{i}$),
ends on any other vertex of $\Vie{i}$ (resp. $\Vio{i}$) and visits all vertices of $\Vie{i}$ (resp. $\Vio{i}$) exactly once using only walks of length $2$, $4$, or $6$.
To achieve this, we first have to build the tree $T_i^{\mathrm{even}}$ (resp. $T_i^{\mathrm{odd}}$) by taking the vertices of $\Vie{i}$ (resp. $\Vio{i}$) and connecting them if they are at distance 2 in $T_i$.
Notice that $T_i^{\mathrm{even}}$ and $T_i^{\mathrm{odd}}$ are indeed trees, where the parent of each vertex is its
grandparent in~$T_i$.
Recall now that we saw in \cref{prop:123hcycle}
that any tree admits a \Hpath[1,2,3] between any 2 given vertices.
We can use this result to create a \Hpath[1,2,3] in $T_i^{\mathrm{even}}$ (resp. $T_i^{\mathrm{odd}}$) between any two of its vertices.
By taking the same vertices in the same order as the path we just constructed, we will end up with a simple $\{2,4,6\}$-path in $T_i$ between any 2 given vertices because every distance is just multiplied by 2.

By using this property, we define for each $i$ a simple $\{2,4,6\}$-path taking all the vertices of $\Vie{i}$ except for $v_i$, denoted $\Pie{i}$, this way:
If $\Vie{i} = \set{v_i}$, $\Pie{i}$ is empty.
Else, if $\Vie{i}$ contains a vertex at depth $4$ in $T_i$, $\Pie{i}$ starts at any vertex in $\Vie{i}$ at depth $4$ in $T_i$ and ends at any vertex at depth $2$ in $T_i$.
Finally, the last case is when $\Vie{i}$ contains only $v_i$ and vertices at depth $2$ in $T_i$, in which case we define $\Pie{i}$ as a path between any two vertices at depth $2$ in $T_i$.
Similarly, we can define for each $i$ a simple $\{2,4,6\}$-path $\Pio{i}$ that
takes all the vertices of $\Vio{i}$ in this way:
If $\Vio{i}$ is empty, $\Pio{i}$ is empty.
Else, if $\Vio{i}$ contains a vertex at depth $3$ in $T_i$, $\Pio{i}$ starts at any vertex in $\Vio{i}$ at depth $1$ in $T_i$ and ends at any vertex at depth $3$ in $T_i$.
Finally, the only other case is when $\Vio{i}$ only contains vertices at
depth $1$ in $T_i$, and then $\Pio{i}$ is a path between any two vertices at depth $1$ in $T_i$.

\medskip

\textbf{Step 3: Creating the cycle.}
We now construct the \Hcycle[2,4,6] by interleaving visits to the cycle $C$ and the Hamiltonian paths that we constructed for $\Vie{i}$ and $\Vio{i}$ within each tree.
Informally, we will go around the cycle twice, the first time visiting every vertex \(v_i\) with even index \(i\)
(skipping the odd indexes). During the second traversal we will
skip the vertices visited in the first traversal and visit the vertices with odd index.
When stopping on a $v_i$, we enumerate $\Pie{i}$, and when skipping a $v_i$, we enumerate $\Pio{i}$.
This allows us to visit every vertex exactly once, and end up on the starting vertex.
The formal construction of the cycle is as follows:

\begin{enumerate}
    \item Start at $v_0$.
    \item For $i = 0$ to $k-1$:
    \begin{enumerate}
        \item Visit the vertices from $\Pie{2i}$, 
          starting at a vertex at depth $2$ or $4$ in $T_{2i}$ (or none if
          $\Pie{2i}$ is empty) and 
          ending on a vertex at depth $2$ in $T_{2i}$
        (or at $v_{2i}$ itself if $\Pie{2i}$ is empty).
        \item If $\Pio{2i+1}$ is not empty, go to its first vertex, which is at
          depth $1$ in $T_{2i+1}$, through a walk of length $4$ (or $2$ if
          $\Pie{2i}$ was empty). Then,
        traverse $\Pio{2i+1}$, ending on a vertex at depth $1$ or $3$ in $T_{2i+1}$.
        \item Move to $v_{2i+2}$. If $\Pio{2i+1}$ was empty, this is a walk of length $2$ through $v_{2i} - v_{2i+1} - v_{2i+2}$.
        Otherwise, this is a walk of length $2$ or $4$ from the last vertex of $\Pio{2i+1}$.
    \end{enumerate}
    \item We can now go through $\Pie{2k}$, ending on a vertex at depth $2$ in $T_{2k}$, or stay at $v_{2k}$ if $\Pie{2k}$ is empty.
    Once this is done, we have visited all vertices of the $\Vie{i}$ for every even $i$ and all vertices of the $\Vio{i}$ for every odd $i$.
    \item We then go to the first vertex of $\Pio{0}$ which is at depth $1$ in $T_0$, through a walk of length $4$ (or $2$ if $\Pie{2k}$ was empty) and traverse $\Pio{0}$ ending at a vertex at depth $1$ or $3$ in $T_0$.
    \item 
    We now just have to repeat the previous loop while swapping the roles of
    even and odd. Formally, for $i = 0$ to $k-1$:
    \begin{enumerate}
        \item Move to $v_{2i+1}$. This is a walk of length $2$ or $4$ from
          either the last vertex of $\Pio{2i}$ or from $v_{2i}$ if $\Pio{2i}$
          was empty.
        \item If $\Pie{2i+1}$ is not empty, visit the vertices of $\Pie{2i+1}$ in order, ending on a vertex at depth $2$ in $T_{2i+1}$. Else stay at $v_{2i+1}$.
        \item Go to the first vertex of $\Pio{2i+2}$ which is at depth $1$ in $T_{2i+2}$ through a walk of length $4$ from the last vertex of $\Pie{2i+1}$ (or of length $2$ from $v_{2i+1}$ if $\Pie{2i+1}$ was empty) and traverse $\Pio{2i+2}$, ending at a vertex at depth $1$ or $3$ in $T_{2i+2}$.
    \end{enumerate}
    \item Finally, return to $v_0$ from the last vertex of $\Pio{2k}$ or from
      $v_{2k}$, which is a walk of length $2$ or~$4$. This completes the cycle
      and concludes the proof. \qedhere
\end{enumerate}
\end{proof}

We can now prove \cref{thm:246hcycle}, which concludes the proof of~\cref{thm:main}:

\begin{proof}[Proof of \cref{thm:246hcycle}]
    Given a graph $G$, we first determine in linear time whether $G$ is bipartite.
    If $G$ is bipartite, then every walk of length in $S = \{2,4,6\}$ will keep
    us in the same part of the bipartition, so we cannot hope to ever visit the
    other side. Otherwise, \cref{lem:246hcyclenb} gives us a \Hcycle[2,4,6] in linear time.
\end{proof}

\section{Variants}\label{sec:variants}
In this section we
consider variants of the \Hcycle{} problem studied in \cref{thm:main}.
We start with the case of
infinite sets $S$, for which we can easily show that the \Hcycle{} problem is either
trivial or solvable in linear time. We then move from cycles to paths,
and discuss the complexity of
the \Hpath{} problem and of the \Hpath{} problem with specified endpoints:
we talk in more detail about the case $S = \{1,2,4\}$ for paths with specified
endpoints, which requires some new reasoning.
Last, we discuss the question of whether we can 
make \Hcycle{} tractable
restricting the class of input graphs.

\subsection{Infinite Sets}\label{sec:infinite}

In this section, we discuss the complexity of the \Hcycle{} problem for infinite
sets~$S$. The crucial observation is that most infinite sets~$S$
closed under subtraction of~$2$ are trivial, because they contain either
$\{1,2,3\}$ or $\{1,2,4\}$. (We give the formal details about this classification in
Appendix~\ref{apx:infinite}.)
The two non-trivial sets are the set of all even integers, and the set of all odd integers.
The case of the set of all even integers follows easily from \cref{thm:246hcycle}.
As for the set of all odd integers, we give its characterization below:

\begin{toappendix}
    \subsection{Proofs for Infinite Sets}\label{apx:infinite}
Let $S$ be an infinite set of positive integers that is closed under subtraction of $2$.
Then, if $S = \NN^+$, then $S$ contains $\{1,2,3\}$ and the problem is trivial by \cref{prop:123hcycle}.
Otherwise, either $S$ contains all the even numbers, or $S$ contains all the odd numbers (or both).
In the first case, then $S$ may also contain some odd numbers, but in
particular it must contain $1$, so that $S$ contains $\{1,2,4\}$, which means that every graph admits an \Hcycle{} by \cref{thm:124hcycle}.
In the second case, then $S$ may also contain some even numbers, but in
particular it must contain $2$, so that $S$ contains $\{1,2,3\}$, which means that every graph admits an \Hcycle{} by \cref{prop:123hcycle}.
Hence, the only non-trivial cases are when $S$ is the set of all even numbers, and when $S$ is the set of all odd numbers.

First, for $S$ the set of all even numbers, the characterization essentially follows
from \cref{sec:poly}. Indeed, as $S$ contains $\{2,4,6\}$, we know by
\cref{lem:246hcyclenb} that every non-bipartite graph has an \Hcycle{} and that we can
compute it in linear time. Now, it is easy to see that bipartite graphs cannot
have an \Hcycle{} for similar reasons to the proof of \cref{thm:246hcycle}:
walks of even length cannot make us switch from one part of the bipartition to the other.
Hence, we have shown:

\begin{theorem}\label{thm:infiniteevenhcycle}
    Let $S$ denote the set of all even integers, and let $G$ be a graph.
    Then $G$ admits an \Hcycle{} if and only if $G$ is non-bipartite.
\end{theorem}

Then, for $S$ the set of all odd numbers, the classification does not follow
from our earlier results. We will now show the following:
\end{toappendix}
\begin{theoremrep}\label{thm:infiniteoddhcycle}
    Let $S$ denote the set of all odd integers, and let $G$ be a graph.
    Then $G$ admits an \Hcycle{} if and only one of the two following cases
    holds: $G$ is non-bipartite,
    or $G$ is bipartite with parts of equal cardinality.
\end{theoremrep}

\begin{proofsketch}
    Let $G$ be a graph.
    If $G$ is non-bipartite, then we can connect any two vertices by a walk of odd length.
    Indeed, we can go through an odd cycle as many times as needed to adjust the parity of the length of that walk.
    Hence, any ordering of the vertices of $G$ is an \Hcycle{}.

    If $G$ is bipartite, the only odd-length walks are the ones connecting vertices from one part to the other.
    Hence, in order to have an \Hcycle{}, the ordering of the vertices must alternate
    between the two parts, which implies that the two parts must have equal size.
    Conversely, if the two parts have equal size, then any ordering that
    alternates between the two parts is an \Hcycle{} because any pair of vertices from different parts
    are connected by a walk of odd length.
\end{proofsketch}

\begin{proof}
    Let $G$ be a graph.
    We will first show that $G$ admits an \Hcycle{} if $G$ is non-bipartite.

    Suppose $G$ is non-bipartite. Then let $C = (v_1, \dots, v_n)$ be any ordering of $V(G)$.
    We claim $C$ is an \Hcycle{} of $G$.
    Indeed, since $G$ is non-bipartite, it contains an odd cycle $O$. 
    For any two consecutive vertices $v_i$ and $v_{i+1}$ in $C$, we will show there exists an odd-length walk between them in $G$.
    Let $P_1$ be any walk from $v_i$ to some vertex $u \in V(O)$, and let $P_2$
    be any walk from $u$ to $v_{i+1}$.
    If $|P_1| + |P_2|$ is odd, then $P_1  P_2$ (the concatenation of \(P_1\) with \(P_2\)) is a walk of odd length between $v_i$ and $v_{i+1}$.
    If $|P_1| + |P_2|$ is even, then traverse $O$ once between $P_1$ and $P_2$, and $P_1 O P_2$ is a walk of length $|P_1| + |O| + |P_2|$, which is odd.
    Since any two consecutive vertices in $C$ are connected by an odd-length walk in $G$, $C$ is an \Hcycle{} of $G$.

    Now it remains to show that if $G$ is bipartite, then $G$ has an \Hcycle{} if and only if its bipartition has equally-sized parts.

    Suppose $G$ is bipartite with parts $A$ and $B$. First, we show that if $G$ admits an \Hcycle{}, then $|A| = |B|$.
    Let $C = (v_1, \dots, v_n)$ be an \Hcycle{} of $G$. In a bipartite graph, a
    walk between two vertices has even length (resp., odd length) if and only if
    it joins two vertices that are in the same parts (resp., in different
    parts).
    Since each consecutive pair $(v_i, v_{i+1})$ in $C$ must be connected by an odd-length walk,
    they must belong to different parts. This means $C$ must alternate
    between $A$ and $B$, so we must have $|A| = |B| = n/2$.

    Conversely, suppose $G$ is bipartite with equally-sized parts $A$ and $B$, where $|A| = |B| = n/2$.
    Consider any ordering $C = (a_1, b_1, a_2, b_2, \dots, a_{n/2}, b_{n/2})$ that alternates between vertices of $A$ and $B$.
    Since consecutive vertices in $C$ belong to different parts, and because
    any two vertices in different parts are connected by a walk which has
    odd length, we know that there is
    indeed an odd-length walk between each consecutive pair. Therefore, $C$ is an \Hcycle{} of $G$.
\end{proof}

 \subsection{Hamiltonian Path Variants}\label{sec:pathandse}

In this section, we discuss the complexity of the two variants of the \Hcycle{}
problem introduced in \cref{sec:prelim}: the \Hpath{} problem, and the \Hpath{} problem with specified endpoints.
The complexity of these problems on the various possible finite sets~$S$ does
not always follow from our results on the \Hcycle{} problem, so we explain the
situation in this subsection and in the next subsection.

First, for any set~$S$, it is clearly always possible to reduce from the
\Hpath{} problem or
from the \Hcycle{} problem to the \Hpath{} problem with specified endpoints, under Cook reductions.
The former is achieved by calling the oracle for \Hpath{} with specified
endpoints on all pairs of vertices, and the latter is achieved by calling that
oracle on all pairs of vertices connected by a walk having a length in~$S$.
However, some of the proofs presented in this paper can be modified to show a better complexity for the \Hpath{} problems.

As most of these problems behave the same way as their \Hcycle{} counterpart, we
will highlight some cases here: we 
present the detailed explanations of every set for both variants (and their proofs
when applicable) in Appendix~\ref{apx:pathandse}.
The only result that has been proved to be different is for the \Hpath[1,2,4]
with specified endpoints. 
Whereas \Hcycle[1,2,4] was trivial (\cref{thm:124hcycle}), for \Hpath[1,2,4] with specified
endpoints the problem can be solved in linear time by excluding a small explicit
family of negative instances. We present this result in the next section.

The only other sets $S$ where the complexities of the path and cycle problems
are not known to coincide
is $S = \{2,4\}$ and $S = \{2,4,6\}$, where we leave some problems
open: see Appendix~\ref{apx:pathandse} for more details.

\begin{toappendix}
\subsection{Proofs for Hamiltonian Path Variants}\label{apx:pathandse}

In this section, we detail the status of the path versions of our main problem
(with and without specified endpoints), for all sets $S$, starting with the known results:
\begin{itemize}
  \item The \Hpath[1] problem, with or without specified endpoints, is well-known to be NP-hard through an easy reduction from the \Hcycle[1] problem.
  \item Concerning the \Hpath[1,2] problems, it is not directly a known result but it is easy to see that the reduction proposed for the \Hcycle[1,2] problem in~\cite{On_graphs_with_Underg_1978} also
    allows us to reduce the \Hpath[1] problem (resp., the \Hpath[1] problem with
    specified endpoints) to the \Hpath[1,2] problem (resp., the \Hpath[1,2]
    problem with specified endpoints).
  \item As for the \Hpath[1,2,3] problems, as we reviewed in
    \cref{prop:123hcycle}
    the proof from~\cite{karaganis1968cube,sekanina1960ordering} states that every graph has a \Hpath[1,2,3] starting and ending
at any pair of vertices, which means that both problems are trivial.
\end{itemize}
Now, we talk about what can be derived from our new results:
\begin{itemize}
  \item 
Because we needed it in \cref{sec:2}, we already proved that the \Hpath[2]
problem with specified endpoints is NP-hard.
As for the variant without specified endpoints, we will prove it to also be NP-hard in \cref{thm:2hpath}.
\item 
As every graph contains a \Hcycle[1,2,4], it must also be the case that
it contains a \Hpath[1,2,4],
simply obtained by breaking the cycle.
However, the situation for specified endpoints \Hpath[1,2,4] is quite different. Indeed, we show in
\cref{thm:124sehpath} that this decision problem is non-trivial but that
it is actually solvable in linear time.
\item The characterization of the \Hpath[2,4,6] problem is exactly the same as the one for the \Hcycle[2,4,6] problem, which is that either
the given graph is bipartite and thus it cannot contain such a path, or it is not and then it must contain a \Hcycle[2,4,6] which is also a \Hpath[2,4,6] of that graph.
Concerning the specified endpoints version, we believe it is possible to modify
the algorithm of the proof of \cref{thm:246hcycle} to be able to start and end
anywhere (so that the problem would be solvable in linear time, and the path with specified
endpoints would exist precisely on non-bipartite graphs, for any pair of
endpoints in that case),
but we leave this problem open.
\item The last problems we are left with are the \Hpath[2,4] problems (with and without
  specified endpoints).
  For the \Hpath[2,4] problem with specified endpoints, the problem admits a Cook
  reduction from the \Hcycle[2,4] problem, so it is NP-hard under Cook
  reductions. For the \Hpath[2,4] problem without specified endpoints, we leave
  its
  complexity open, but we conjecture that it is NP-hard like
  \Hcycle[2,4].
\item Finally, the \Hpath[1,3] problems (with or without prescribed endpoints),
  and the problems with finite sets $S = \{1, \ldots, 2k+1\}$ of odd
  integers for some $k>0$, are left open by the present work, like the corresponding \Hcycle{}
  problem.
\end{itemize}
\end{toappendix}

 \begin{toappendix}
\subsection{Proofs for the \safebm{$\{2\}$}-Hamiltonian Path Problem}\label{apx:2hpath}

We prove in this appendix that the \Hpath[2] problem is NP-complete, by a reduction from
the \Hpath[2] problem with specified endpoints.

\begin{theorem}\label{thm:2hpath}
    The \Hpath[2] problem is NP-complete.
\end{theorem}

\begin{proof}
    We reduce from the \Hpath[2] problem with specified endpoints which we
    proved to be NP-complete in \cref{thm:2sehpath}.

    Let $G$ be a graph and $\alpha, \beta \in V(G)$ two different vertices of $G$.
    Let us construct in polynomial time a graph $H$ and show that $G$ has a \Hpath[2] between $\alpha$ and $\beta$ if and only if $H$ contains a \Hpath[2].

    We construct $H$ by connecting $G$ to two copies of the gadget $D$ represented in \cref{fig:gadget2hpath}
    by adding an edge from $\alpha$ to the vertex $d_0$ in one copy, and an edge
    from $\beta$ to the vertex $d_0$ in the other copy.

    \begin{figure}
        \centering
        \begin{tikzpicture}
            \node[circle, draw] (0) at (0,0) {$d_0$};
            \node[circle, draw] (1) at (-2,0) {$d_1$};
            \node[circle, draw] (2) at (-1,2) {$d_2$};
            \node[circle, draw] (3) at (2,0) {$d_3$};
            \node[circle, draw] (4) at (4,0) {$d_4$};
            
            \draw (0) -- (1);
            \draw (0) -- (2);
            \draw (2) -- (1);
            \draw (3) -- (4);
            \draw (0) -- (3);
        \end{tikzpicture}
        \caption{Dangling gadget $D$ for the proof of \cref{thm:2hpath}}\label{fig:gadget2hpath} 
    \end{figure}

    Formally, $H$ is constructed by copying $G$ and adding $D$ and $D'$ to $G$ (with $D'$ a copy of $D$ such that $V(D') = \set{d_0',d_1',d_2',d_3',d_4'}$),
    which we will connect to~$G$ with the edges $(\alpha, d_0)$ and $(\beta, d_0')$.

    \medskip

    We first show the forward direction of the reduction.
    Let us assume that $G$ has a \Hpath[2] $P$ starting at $\alpha$ and ending at $\beta$.
    Then one can easily verify that the path $P' = (d_4, d_0, d_2, d_3, d_1) P (d_1', d_3', d_2', d_0', d_4')$
    is a \Hpath[2] of $H$.

    \medskip

    We now show the more interesting backward direction.
    Suppose $H$ has a \Hpath[2] $Q$.
    Because $d_4$ and $d_4'$ only have one vertex at distance 2 (resp., $d_0$ and
    $d_0'$), they must be the endpoints of $Q$. Without loss of generality, up to
    reversing $Q$, suppose $d_4$ is the start and $d_4'$ is the end.
    From $d_4$ we must go to $d_0$ and now we have the possibility of going out of $D$.
    If we do, then the vertices $d_1,d_2,d_3$ of $D$ are not yet
    visited, and the only vertex of $G$
    to which they are connected by a walk of length $2$ is~$\alpha$. This means we can only enter $D$ again by using $\alpha$
    which prevents us from going out afterwards and implies that we cannot
    reach $d_4'$ which must be at the end of $Q$.
    For this reason, we need to visit all vertices of \(D\) before going to \(\alpha\).
    We will now eventually leave $D$ and it will necessarily be via $\alpha$ and we
    will not re-enter $D$ again afterwards.
By a symmetric reasoning, we can only take vertices of $D'$ after $\beta$ and all vertices of $D'$ must be after $\beta$.

    In summary, we have shown that $Q$ is composed of three $\{2\}$-paths:
    \begin{itemize}
      \item One from $d_4$ to $\alpha$ by taking only vertices of $D$;
      \item one from $\alpha$ to $\beta$ by taking only vertices of $G$; 
      \item and one from $\beta$ to $d_4'$ by taking only vertices of $D'$.
      \end{itemize}
    Then by removing the vertices before $\alpha$ and those after $\beta$ we get a
    \Hpath[2] from $\alpha$ to $\beta$ in~$G$.

    This completes the proof that the \Hpath[2] problem is NP-complete.
\end{proof}
\end{toappendix}

 \subsection{\texorpdfstring{\set{1,2,4}-Hamiltonian Path}{\{1,2,4\}-Hamiltonian Path} Problem with Specified Endpoints}\label{sec:124sehpath}

\begin{toappendix}
\subsection{Proofs for the \texorpdfstring{\set{1,2,4}-Hamiltonian Path}{\{1,2,4\}-Hamiltonian Path} Problem with Specified Endpoints}
\end{toappendix}

In this section, we prove that the \Hpath[1,2,4] problem with specified
endpoints can be solved in linear time.
For those proofs, we will need to introduce the notions of \emph{star} and
\emph{bistar}.
A \emph{star centered at a vertex $x$} is a graph with at least two vertices such that every vertex different from $x$ is adjacent only to $x$.
A \emph{bistar centered at vertices $a$ and $b$} is a graph such that $a$ and $b$ are adjacent and every vertex different from $a$ and $b$ is
adjacent to only either $a$ or $b$ (but not both), and both $a$ and $b$ have degree at least $2$. These two definitions are illustrated in \cref{fig:star}.
We now introduce a lemma that will help us prove the desired theorem.

\begin{figure}[ht]
    \centering
    \begin{tikzpicture}

\node[circle, draw] (x) at (0,0) {$x$};

        \foreach \i/\angle in {1/90,2/18,3/-54,4/-126,5/162} {
            \node[circle, draw] (s\i) at (\angle:2) {$v_{\i}$};
            \draw (x) -- (s\i);
        }

\node[circle, draw] (a) at (6,0) {$a$};
        \node[circle, draw] (b) at (8,0) {$b$};
        \draw (a) -- (b);

\node[circle, draw] (a1) at (5.2,1.2) {$a_1$};
        \node[circle, draw] (a2) at (5.2,-1.2) {$a_2$};
        \node[circle, draw] (a3) at (4.6,0) {$a_3$};
        \draw (a) -- (a3);
        \draw (a) -- (a1);
        \draw (a) -- (a2);

\node[circle, draw] (b1) at (8.8,1.2) {$b_1$};
        \node[circle, draw] (b2) at (8.8,-1.2) {$b_2$};
        \draw (b) -- (b1);
        \draw (b) -- (b2);

    \end{tikzpicture}
    \caption{A star centered at $x$ (left) and a bistar centered at $a$ and $b$ (right).}\label{fig:star}
\end{figure}

\begin{lemmarep}\label{lem:124sehpathtree}
    Let $T$ be a tree with $\lvert V(T) \rvert > 2$ and $a,b \in V(T)$ two vertices of $T$ such that $a$ and $b$ are not adjacent.
    Then there exists a \Hpath[1,2,4] between $a$ and $b$. 
\end{lemmarep}

\begin{proofsketch}
    The proof is done by induction on the number of vertices of $T$.
    The idea is to divide the tree into multiple vertex disjoint subtrees depending on the shape of $T$,
    and to use the induction hypothesis as well as \cref{lem:124hcycletree} to construct simple $\{1,2,4\}$-paths
    of these subtrees, which can then be concatenated to form the desired \Hpath[1,2,4] between $a$ and $b$.
\end{proofsketch}

\begin{proof}
    We will prove this result by induction on the number of vertices of $T$.

    \textbf{Base case.}
    There is only one tree of size $3$,
which is the path of size $3$.
    The only two vertices that are not adjacent are the vertices at both ends of
    the path, and we can achieve a \Hpath[1,2,4] between them by taking two steps of length~$1$.

    \textbf{Induction step.}
    Assume that the statement holds for all trees with strictly fewer than $n \ge 4$ vertices.
    Let $T$ be a tree with $|V(T)| = n$, and let $a,b \in V(T)$ be two non-adjacent vertices.

    Because $T$ is a tree, there is a unique simple $\{1\}$-path between $a$ and
    $b$, which we write $P_{ab} = (a,v_1, \ldots, v_k,b)$.
    Because $a$ is not adjacent to $b$, we have $k \ge 1$. 
    Let $T'$ be the forest obtained from $T$ by removing every edge of $P_{ab}$,
    and let $T_a$ be the connected component of $T'$ containing $a$.

    If $k \ge 2$, then we can define two paths $P_a$ and $P_b$ that will form a \Hpath[1,2,4] of $T$ when concatenated.
    If $a$ is a leaf, then we define $P_a$ as $(a)$. Else, $T_a$ has at least two vertices and 
    we use \cref{lem:124hcycletree} to get a \Hcycle[1,2,4] of $T_a$ $(a,a_1,\dots,a_m)$ where $a_1$ is adjacent to $a$.
    We then define $P_a$ as $(a,a_m,\dots,a_1)$.
    By removing every vertex of $T_a$ from $T$, we get a tree $T_{v_1 b}$ in which
    $v_1$ and $b$ are non-adjacent and $|V(T_{v_1 b})| < n$ because at least
    we removed $a$.
    This means we can use the induction hypothesis on $T_{v_1 b}$ to get a \Hpath[1,2,4] from $v_1$ to $b$, which we call $P_b$.
    Concatenating $P_a$ and $P_b$ yields a path from $a$ to $b$ which goes through every vertex of $T$.
    Every step in $P_a$ or $P_b$ is of length $1$, $2$ or $4$ due to their construction.
    Furthermore, the transition between them is either from $a$ to $v_1$ (of length $1$) if $a$ is a leaf,
    or else from $a_1$ to $v_1$ through $a$ (of length $2$). Thus, we have
    constructed a \Hpath[1,2,4] of $T$ from $a$ to $b$.

    If $k = 1$, let $c$ be the common neighbor of $a$ and $b$.
    Let $T_c$ be the connected component of $T'$ containing $c$ and $T_b$ be the one containing $b$.
There are two cases: either $T_c = \set{c}$, or $T_c$ contains at least two
    vertices.
    In the first case, we can define $P_a$ in the same way as above. Namely,
    if $a$ is a leaf, $P_a = (a)$. Alternatively, we use \cref{lem:124hcycletree} to get a \Hcycle[1,2,4] of $T_a$ $(a,a_1,\dots,a_m)$ where $a_1$ is adjacent to $a$
    and we define $P_a = (a,a_m,\dots,a_1)$. We then define $P_b$ similarly.
    Namely, if $b$ is a leaf, $P_b = (b)$. Else we use the same lemma to get a \Hcycle[1,2,4] of $T_b$ $(b,b_1,\dots,b_l)$
    and we can define $P_b$ as $(b_1, \dots, b_l, b)$. Then $P = P_a (c) P_b$ is a \Hpath[1,2,4] of $T$ from $a$ to $b$.
    In the first case,
we have $|V(T_c)|\geq 2$ and we can use \cref{lem:124hcycletree} to get a \Hcycle[1,2,4] of $T_c$ $(c,c_1,\dots,c_j)$ where $c_1$ is adjacent to $c$.
    If both $a$ and $b$ are leaves, then $P = (a, c, c_j, \dots, c_1, b)$ is a \Hpath[1,2,4] of $T$ from $a$ to $b$.
    Otherwise, if both $T_a$ and $T_b$ are stars centered at $a$ and $b$ respectively, such that $V(T_a) = \set{a,a_1,\dots,a_m}$ and $V(T_b) = \set{b,b_1,\dots,b_l}$,
    then $P = (a,a_1,\dots,a_m, b_1, \dots, b_l, c,c_j, \dots, c_1, b)$ is a \Hpath[1,2,4] of $T$ from $a$ to $b$.
    Finally, if one of $T_a$ or $T_b$ is not a star centered at $a$ or $b$, then one of them must have a vertex at distance $2$ from $a$ or $b$.
    Without loss of generality, suppose $T_a$ has a vertex $a_1$ at distance $2$
    from $a$. Then we can use the induction hypothesis on~$T_a$ to get a \Hpath[1,2,4] $P_a$ from $a$ to $a_1$.
    We now define $P_b$ exactly as above and $P = P_a (c_1, \dots, c_j, c) P_b$ is a \Hpath[1,2,4] from $a$ to $b$.
\end{proof}

We can now show our characterization for the \Hpath[1,2,4] problem with
specified endpoints:

\begin{theoremrep}\label{thm:124sehpath}
    Let $G$ be a graph and $a$ and $b$ be two different vertices of $G$.
    Then $G$ admits a \Hpath[1,2,4] from $a$ to $b$ if and only if $G$ is not a bistar centered at $a$ and $b$.
\end{theoremrep}

\begin{proofsketch}
    Let $G$ be a graph and $a$ and $b$ be two different vertices of $G$.
    If they are non-adjacent, we conclude by \cref{lem:124sehpathtree}. If $a$
    and $b$ are adjacent and $G$ is not a bistar, then 
    up to exchanging $a$ and $b$ we find a neighbor $v_1$ of $a$ which is not $b$ and
    which has itself a neighbor $v_2$ which is also not $b$. We take a spanning
    tree of~$G$ including the edges $\edge{a,v_1}$ and $\edge{v_1,v_2}$, and we
    use
\cref{lem:124sehpathtree} 
    and finally obtain a \Hpath[1,2,4] of $G$.

    Now, if $G$ is a bistar, a path from $a$ can start from neighbors of $a$ (resp. neighbors of $b$) but will not be able to reach neighbors of $b$ (resp. $a$) without taking $b$ (which must be at the end).
    Thus a bistar centered at $a$ and $b$ cannot contain a \Hpath[1,2,4] from $a$ to $b$, which concludes the proof.
\end{proofsketch}

\begin{proof}
    Let $G$ be a graph and $a$ and $b$ be two different vertices of $G$.
    If $a$ and $b$ are not adjacent in $G$, then let $T$ be a spanning tree of $G$.
    In $T$, $a$ and $b$ are also not adjacent. Thus, we can use \cref{lem:124sehpathtree} to get a \Hpath[1,2,4] of $T$ from $a$ to $b$,
    which is also a \Hpath[1,2,4] of $G$ from $a$ to $b$.

    Suppose $a$ and $b$ are adjacent and that $G$ is not a bistar centered at $a$ and $b$,
    then there exists two vertices $v_1, v_2 \in G$ different from $a$ or $b$ such that $v_2$ is adjacent to $v_1$ and $v_1$ is adjacent to $a$ or $b$.
    Without loss of generality, suppose $v_1$ is adjacent to $a$.
    Then let $T$ be a spanning tree of $G$ such that \edge{a,v_1} and \edge{v_1,v_2} are edges of $T$.
    Let $T'$ be the tree obtained from $T$ by removing the edge \edge{a,b}, let $T_a$ be the connected component of $T'$ containing $a$ and $T_b$ the one containing $b$.
    Then we can use \cref{lem:124sehpathtree} to get a \Hpath[1,2,4] $P_a$ from $a$ to $v_2$ in $T_a$.
    We can also use \cref{lem:124hcycletree} to get a \Hcycle[1,2,4] $(b,b_1,\dots,b_l)$ where $b$ and $b_1$ are adjacent.
    Then $P = P_a (b_1, \dots, b_l, b)$ is a \Hpath[1,2,4] of $T$ and of $G$ from $a$ to $b$.

    Finally, we must show that a bistar cannot contain a \Hpath[1,2,4]
    connecting its two centers.
    Suppose $G$ is a bistar centered at $a$ and $b$ and let $(a_1, \dots, a_m)$
    with $m \geq 1$ be the neighbors of $a$ that are different from $b$, and
    let $(b_1, \dots, b_l)$ with $l \geq 1$ be the neighbors of~$b$ that are different from $a$.
    From $a$ we can reach any vertex. If our path starts by going from $a$ to
    some $a_i$ for some $i$, then the only vertices it can reach from now on are either other $a_j$ for $j \neq i$ or $b$.
    Going to another $a_j$ does not change that fact which means there will be
    no way to reach a $b_i$ without taking $b$. Because $b$ must be taken last,
    we have excluded the case of a path that starts with some $a_i$.
    However, the only other solution is to start by taking a $b_i$, and this
    does not work for the exact same reasoning: we cannot visit the $a_j$
    without visiting $a$.
    Thus there is no \Hpath[1,2,4] from $a$ to $b$ in $G$.
\end{proof}

 \subsection{Restricted Graph Classes}\label{sec:restricted}

\begin{toappendix}
\subsection{Proofs for Bounded Cliquewidth}\label{apx:boundedcliquewidth}
\end{toappendix}

In this section, we last turn to one last question: can the \Hcycle{} problem
be made tractable if we assume that the input graphs are in restricted graph
classes? This is motivated by the fact that the Hamiltonian cycle problem, while
NP-hard in general, is known to be tractable on many restricted classes of
graphs, in particular 
bounded-treewidth and bounded-cliquewidth graphs~\cite{bergougnoux2020optimal, espelage2001solve}. Hence, one can ask
whether the same is true of the \Hcycle{} problem for sets $S$ other than $S = \{1\}$.

There are related results in this direction for the case $S = \{1,2\}$ and when
the input graph is a tree.
Namely, it is known that a tree $T$ has a
\Hcycle[1,2] if and only if it is a so-called caterpillar graph~\cite{Trees_with_Hami_Harary_1971}.
Further, the trees having a \Hpath[1,2] have been characterized
in~\cite{Hamiltonian_Pat_Radosz} and can be recognized in linear time.

We now claim that, for arbitrary finite sets~$S$, 
tractability extends beyond trees to bounded-treewidth and even
bounded-cliquewidth graphs. Namely:

\begin{theoremrep}\label{thm:boundedtw}
  For any fixed finite set $S \subseteq \NN^+$ and bound $k \in \NN$, the \Hcycle{} problem on
  input graphs of cliquewidth at most $k$ is in polynomial time.
\end{theoremrep}

\begin{proofsketch}
The proof uses the fact that bounded cliquewidth is preserved by MSO
interpretations~\cite{bonnet2022model,blumensath2006recognizability}. 
So, from the input graph $G$, we construct the graph $G_S$ that
has the same vertices as $G$ and connects two vertices if and only if they are
connected by a walk of a length in $S$:
this is definable by an MSO interpretation, so $G_S$ also has bounded
cliquewidth.
It is now easy to see that $G_S$ admits a Hamiltonian cycle if and only
if $G$ admits an \Hcycle. From there, thanks to the result from~\cite{bergougnoux2020optimal}, there exists an XP algorithm
to solve the Hamiltonian cycle problem on graphs with bounded
cliquewidth, which allows us to conclude.
\end{proofsketch}

\begin{proof}
Let $S=\{d_{1},\dots,d_{m}\}\subseteq \NN^+$ be the fixed finite set of
authorized lengths. Let $k$ be the constant cliquewidth bound, and let
$\mathcal{C}$ be a class of finite graphs of cliquewidth bounded by~$k$.

We define for each $d_i \in S$ the following FO formula:
$\phi_{d_i}(x,y) = \exists z_2 \cdots \exists z_{d_i} (E(x,z_2) \wedge E(z_2,z_3) \wedge \dots \wedge E(z_{d_i},y))$ with $E(x,y)$ the edge relation of the graph $G$.
We also define the following formula: $\phi(x,y) = \bigvee_{d_i \in S} \phi_{d_i}(x,y)$.
For each $d_i \in S$, the formula $\phi_{d_i}(x,y)$ defines the constraint that
$x$ and $y$ are connected in $G$ by a walk of length $d_i$, and
the formula $\phi(x,y)$ then defines the constraint that $x$ and $y$ are
connected in $G$ by a walk of length $\ell$ for some $\ell \in S$.

This allows us to define the following MSO interpretation:
For $G\in\mathcal{C}$, define a new graph $H$ with the same vertex set
$V(H)=V(G)$, and with edge set $E(H)$ defined by
$E(H) = \left\{\edge{u,v} | (u,v) \in V(G)^2 \mid G\models \phi(u,v) \right\}$.

We now show that $G$ admits an \Hcycle{} if and only if $H$ admits a Hamiltonian cycle.
First, if $G$ admits an \Hcycle{}, then there exists a cycle $C$ in $G$ that visits every vertex exactly once and such that for every two consecutive vertices
$u$ and $v$ in $C$, there exists a walk of length in $S$ between $u$ and $v$. By definition of $H$, this means that there is an edge between $u$ and $v$ in $H$.
Hence, the cycle $C$ is also a Hamiltonian cycle in $H$.
Conversely, if $H$ admits a Hamiltonian cycle, then there exists a cycle $C$ in $H$ that visits every vertex exactly once and such that for every two consecutive vertices
$u$ and $v$ in $C$, there is an edge between $u$ and $v$ in $H$. By definition of $H$, this means that there exists a walk of length in $S$ between $u$ and $v$ in $G$.
Hence, the cycle $C$ is also an \Hcycle{} in $G$.

Because $\mathcal{C}$ has bounded cliquewidth, and because MSO interpretations preserve bounded cliquewidth~\cite{bonnet2022model,blumensath2006recognizability},
letting $\mathcal{C}'$ be the class of the graphs $H$ that can be obtained from
$G\in\mathcal{C}$ by the above MSO interpretation, then $\mathcal{C}'$
also has bounded cliquewidth: let $k'$ be a bound on the cliquewidth of the
graphs in~$\mathcal{C}'$.
Finally, because the Hamiltonian cycle problem is in XP for graphs of bounded
cliquewidth~\cite{bergougnoux2020optimal}, there exists an algorithm that solves
the \Hcycle{} problem on $\mathcal{C}$ in XP time, defined in the following way.
Given a graph $G \in \mathcal{C}$, we construct the graph $H$ as defined above,
which takes polynomial time since $S$ is fixed, and $H$ must have cliquewidth
bounded by $k'$.
Using the result from~\cite[Theorem~1.1]{oum2006approximating}, we can get a decomposition of $H$ of width at most $f(k')$ for some function $f$.
Then, using the result from~\cite[Theorem~1]{bergougnoux2020optimal}, we can solve the Hamiltonian cycle problem on $H$ in time $O(n^{g(f(k'))})$ for some function $g$.

This gives us an algorithm to solve the \Hcycle{} problem on $\mathcal{C}$ in time $O(n^{g(f(k'))})$, which is polynomial since $k'$ is a constant depending only on $S$ and $k$.
\end{proof}

However, bounded cliquewidth is not the end of the story for the Hamiltonian cycle problem, as it is also known to be tractable on some classes of unbounded cliquewidth,
e.g., proper interval graphs~\cite[Theorem~5]{A_simple_algori_Ibarra_2009}.
Hence, a natural question for future work would be to explore whether the \Hcycle{} problem is also tractable on other classes of unbounded cliquewidth.

\section{Conclusion and Future Work}
\label{sec:conc}
We have studied the \Hcycle{} problem and we determined its 
complexity for most finite sets $S$. We classified the complexity of every result into a
decision tree that allows to restrict our attention to fewer cases. Among those cases, some were known results,
and others have been proven in this paper, namely the NP-completeness when $S = \set{2}$ and $S = \set{2,4}$, the triviality when $S = \set{1,2,4}$,
and the polynomial-time solvability when $S = \set{2,4,6}$. Furthermore, we also
addressed the case of infinite sets $S$, which are either trivial or
polynomial-time solvable;
and
we also studied the variants of paths, with or without specified endpoints, for which we obtained similar complexity classifications for most cases.

We now mention potential future work.
The most immediate open question is to classify the complexity of the \Hcycle{} problem for $S = \set{1,3}$
(and of other finite sets of odd integers): we suspect that this problem is NP-hard, but we have not been able to show it.
However, it is important to notice that in case of bipartite graphs (and for
the same reasons as in \cref{thm:infiniteoddhcycle} on the infinite set of all
odd numbers),
a necessary condition to admit an \Hcycle{} for a finite set $S$ of odd
integers
is to have both parts of the same size. That is nonetheless not a
sufficient condition,
as there may be an ``imbalance'' of vertices of a certain part of the bipartition
in some local part of the graph, that hinders the alternation of
parts in our cycle. Yet, we do not know how to conclude from here, and the
complexity of \Hcycle{} remains open, even when the input graph is restricted to
be bipartite, or to be non-bipartite.
Concerning the path variants of our problems, obviously the same as above holds true for $S$ a finite set of only odd numbers,
but there are also the cases of \Hpath[2,4] and 
\Hpath[2,4,6] with specified endpoints that we have left open.

Another direction would be to study the alternate definition that
does not allow back-and-forth travels in intermediate walks of $S$-paths.
This definition is harder to grasp, and there are a lot more sets to
consider because we can no longer assume without loss of generality
that $S$ is closed by subtraction of $2$. One could also consider the
more stringent definition that requires the back-and-forth travels to form
simple paths, or those where the lengths in~$S$ are required to correspond to
distances between the two vertices (i.e., prohibiting the existence of
shorter paths).

Last, another natural question would be to
look into directed graphs. However, it seems that for any given
finite set $S$, there exist arbitrarily large graphs that do not
contain an \Hcycle. This suggests that the general complexity
answers for directed graphs might be totally different from those on
undirected graphs.

\bibliography{references}

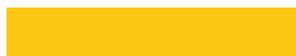
\end{document}